\documentclass[11pt]{article}
\usepackage{enumitem}
\usepackage[font=small,format=plain,labelfont=bf,up]{caption}
\usepackage[a4paper,nohead,ignorefoot,%
twoside,scale=0.852,
twoside,scale=0.86, %
bindingoffset=1.25cm]{geometry}

\usepackage{amssymb,amsfonts,amsmath,amsthm}

\usepackage[]{graphicx}
\usepackage{color}

\graphicspath{{}}
\newcommand{\url}[1]{#1} % dummy definition of \texorpdfstring
\usepackage{hyperref}%
\definecolor{gray}{rgb}{0.2,0.2,.2}
\hypersetup{%
  unicode=true,          % non-Latin characters in Acrobat’s bookmarks
  colorlinks=true,       % false: boxed links; true: colored links
  linkcolor=gray,          % color of internal links
  citecolor=gray      % color of links to bibliography
}
\definecolor{colorGreen}{rgb}{0.,0.75,0.}
\definecolor{colorBlue}{rgb}{0.,0.,0.75}
\definecolor{colorRed}{rgb}{0.99,0.,0.}
\newcommand{\BMHR}{}
\newcommand{\EMHR}{}

\newcommand{\iu}{\mathtt{i}}
\newcommand{\mhexp}[1]{{{\mathtt{e}}^{#1}}}

\newcommand{\fspace}[1]{{\mathsf{#1}}}
\newcommand{\fspaceL}{\fspace{L}}
\newcommand{\fspaceH}{\fspace{H}}
\newcommand{\fspaceC}{\fspace{C}}

\newcommand{\ol}[1]{{\overline{#1}}}

\newcommand{\Rset}{{\mathbb{R}}}
\newcommand{\Zset}{{\mathbb{Z}}}

\newcommand{\Nset}{{\mathbb{N}}}

\newcommand{\ocinterval}[2]{(#1,\,#2]}%
\newcommand{\cointerval}[2]{[#1,\,#2)}%
\newcommand{\ccinterval}[2]{[#1,\,#2]}%

\newcommand{\loc}{{\rm loc}}

\newlength{\mhpicDwidth}
\newlength{\mhpicDvsep}
\newlength{\mhpicDhsep}

\newlength{\mhpicPwidth}
\newlength{\mhpicPvsep}
\newlength{\mhpicPhsep}

\newlength{\mhpicWhsep}

\newcommand{\pair}[2]{{\left({#1},\,{#2}\right)}}
\newcommand{\skp}[2]{{\left\langle{#1},\,{#2}\right\rangle}}

\newcommand{\bskp}[2]{{\big\langle{#1},\,{#2}\big\rangle}}

\newcommand{\at}[1]{{\left({#1}\right)}}
\newcommand{\nat}[1]{(#1)}
\newcommand{\bat}[1]{{\big(#1\big)}}
\newcommand{\Bat}[1]{{\Big(#1\Big)}}

\newcommand{\triple}[3]{{\left({#1},\,{#2},\,{#3}\right)}}

\newcommand{\ul}[1]{\underline{#1}}

\newcommand{\D}{\displaystyle}

\newcommand{\bigpar}{\par\quad\newline\noindent}

\newcommand{\norm}[1]{\left\|{#1}\right\|}
\newcommand{\bnorm}[1]{\big\|{#1}\big\|}

\newcommand{\nnorm}[1]{\|{#1}\|}
\newcommand{\abs}[1]{\left|{#1}\right|}
\newcommand{\babs}[1]{\big|{#1}\big|}
\newcommand{\Babs}[1]{\Big|{#1}\Big|}

\newcommand{\dint}[1]{\,\mathrm{d}#1}

\newcommand{\Ga}{{\Gamma}}

\newcommand{\al}{{\alpha}}
\newcommand{\be}{{\beta}}

\newcommand{\eps}{{\varepsilon}}

\newcommand{\ka}{{\kappa}}
\newcommand{\la}{{\lambda}}
\newcommand{\si}{{\sigma}}

\newcommand{\calC}{\mathcal{C}}

\newcommand{\calK}{\mathcal{K}}

\newcommand{\calP}{\mathcal{P}}
\newcommand{\calQ}{\mathcal{Q}}

\newcommand{\calT}{\mathcal{T}}

\theoremstyle{plain}
\newtheorem{theorem}{Theorem}[]
\newtheorem{corollary}[theorem]{Corollary}

\newtheorem{lemma} [theorem]{Lemma}
\newtheorem{proposition}[theorem]{Proposition}
\newtheorem{conjecture}[theorem]{Conjecture}

\newtheorem{result}[theorem]{Main result}

\newtheorem{assumption} [theorem]{Assumption}

\numberwithin{figure}{section}
\numberwithin{table}{section}
\numberwithin{equation}{section}
\usepackage{float}
\begin{document}
%
% -----------------------------------------------------------------------------
% - Title information
% -----------------------------------------------------------------------------
%
\title{Nonlinear and Nonlocal Eigenvalue Problems:\\ \emph{variational existence, decay properties,  approximation, and\\ universal scaling limits}}%: Nonlinear eigenvalue problems for convolution operators}
\date{\today}
\author{%
Michael Herrmann%
\footnote{Technische Universit\"at Braunschweig, Germany, {\tt michael.herrmann@tu-braunschweig.de}}
\and
Karsten Matthies%
\footnote{University of Bath, United Kingdom, {\tt k.matthies@bath.ac.uk}}
} %
\maketitle
%
%
%
% -----------------------------------------------------------------------------
% - Abstract
% -----------------------------------------------------------------------------
%
\vspace{-0.025\textheight}
\begin{abstract}
We study a class of nonlinear eigenvalue problems which involves a convolution operator as well as a superlinear nonlinearity. Our variational existence proof is based on constrained optimization and provides a one-parameter family of solutions with
positive eigenvalues and unimodal eigenfunctions. We also discuss the decay properties and the numerical computations of those eigenfunctions, and conclude with two asymptotic results concerning small and large eigenvalues.
\end{abstract}
%
%
% -----------------------------------------------------------------------------
% - MSC and keywords
% -----------------------------------------------------------------------------
%
\quad\newline\noindent%
\begin{minipage}[t]{0.15\textwidth}%
 Keywords:
\end{minipage}%
\begin{minipage}[t]{0.8\textwidth}%
\emph{nonlinear eigenvalue problems}, \emph{nonlocal coherent structures},\\  \emph{asymptotic analysis of nonlinear integral operators}
\end{minipage}%
\medskip
\newline\noindent
\begin{minipage}[t]{0.15\textwidth}%
MSC (2010): %
\end{minipage}%
\begin{minipage}[t]{0.8\textwidth}%
45G10,  % nonlinear integral equations
45M05, % Asymptotics of integral equations
47J10,  	%Nonlinear spectral theory, nonlinear eigenvalue problems
49R05  %	Variational methods for eigenvalues of operators
\end{minipage}%
%
%
% -----------------------------------------------------------------------------
% - Table of contents
% -----------------------------------------------------------------------------
%
\setcounter{tocdepth}{3}
\setcounter{secnumdepth}{3}{\scriptsize{\tableofcontents}}
%
%
%
%----------------------------------------------------------------
\section{Introduction}%\label{sect:Intro}
%----------------------------------------------------------------
%
%
This paper concerns the scalar nonlinear eigenvalue problem
\begin{align}\label{eq:nlEv}
  \sigma\, V = b \ast  f(b \ast V)\,,
\end{align}
with eigenvalue $\si$ and unknown eigenfunction $V$, where $f$ is a superlinear function and $b$ is a sufficiently nice convolution kernel. Equation \eqref{eq:nlEv} is a prototypical example for a huge class of similar problems but it seems that this equation has not yet been studied systematically. We restrict our considerations to the one-dimensional setting, where $b$ and $V$ depend on $x\in\Rset$, and develop a variational existence theory. In particular we prove the existence of a one-parameter family of solutions $\pair{\si}{V}$ with $\si>0$ and $V$ being even, nonnegative, and unimodal, where the latter means increasing and decreasing  for $x\leq0$ and $x\geq0$, respectively. Notice that \eqref{eq:nlEv} implies
\begin{align}\label{eq:nlEv2}
  \sigma\, U = a\ast f(U)
\end{align}
via the identification
\begin{align}
\label{eqn:trafo}
U=b\ast V\,,\qquad a=b\ast b
\end{align}
and that \eqref{eq:nlEv2} can be transformed into the symmetric form \eqref{eq:nlEv} provided that we find a kernel $b$ such that the second identity in \eqref{eqn:trafo} is satisfied.
%
%
%
%----------------------------------------------------------------
\subsection{Examples and application }
%----------------------------------------------------------------
%
%
%
\paragraph{Waves in FPUT chains}  A first application are traveling waves in Fermi-Pasta-Ulam-Tsingou (FPUT) chains, which are determined by the advance-delay-differential equations
\begin{align}
\label{Eqn:TW.Diff}
R^\prime\at{x}
=
V\at{x+1/2}-
V\at{x-1/2}\,,\qquad
\si \, V^\prime\at{x}
=
\Phi^\prime\bat{R\at{x+1/2}}-
\Phi^\prime\bat{R\at{x-1/2}}\,,
\end{align}
see for instance \cite{Her10}. Here, $\Phi^\prime$ describes the atomic forces, $x$ is the space variable in the comoving frame,  $R$ and $V$ are the unknown profile functions for the distances and velocities, respectively, and  $\si>0$ stands for the squared wave speed. After integration with respect to $x$, and ignoring all constants of integration for simplicity, the system \eqref{Eqn:TW.Diff} can be written as
\begin{align}
\notag%\label{Eqn:TW.Intsys}
 R = b \ast V\,,\qquad  \si\, V =b \ast  \Phi'\at{R}\,,
\end{align}
and implies \eqref{eq:nlEv} with $f=\Phi^\prime$ after elimination of $R$. \BMHR In this example, the convolution kernel $b$ is given by \EMHR the characteristic function of the interval $\ccinterval{-1/2}{1/2}$ and $a=b\ast b$ is a tent map, see the second row in Figure \ref{fig:kernels}.  Traveling waves in FPUT chains have been studied intensively over the last decades, see for instance \cite{FW94,FV99,HM17,HM19a} and references therein, and the present paper generalizes some of the methods and techniques that have been developed in this context to a broader class of nonlinear eigenvalue problems.
\paragraph{Nonlocal aggregation models}  Nonlinear eigenvalue problems also arise in certain models for biological aggregation. For instance, \cite{BHW13} proposes (among other models) the nonlocal parabolic PDE
\begin{align}
\label{Eqn:Aggregation1}
\partial_t \varrho=\partial_x^2 \bat{h\bat{a\ast\varrho}\varrho}
\end{align}
for a nonnegative density $\varrho$ depending on time $t>0$ and space $x\in\Rset$. The scalar nonlinear function $h$ -- which is supposed to be \BMHR nonnegative \EMHR and strictly decreasing -- models that the diffusion mobility of a biological species depends on the local density of the population so that aggregation is possible. The equation for a bounded steady state reads $h\bat{a\ast\varrho}\varrho=c$, where $c$ denotes a constant of integration, and can be transformed into \eqref{eq:nlEv2} via $f:=1/h$, $\si=1/c$, and $U:=a\ast\varrho$. The initial value problem for \eqref{Eqn:Aggregation1} has been investigated in \cite{BHW13,HO15} in one and even higher space-dimensions, but the existence and the properties of steady states have not yet been examined.
\par
Besides \eqref{Eqn:Aggregation1} there \BMHR exist \EMHR other nonlocal models for biological aggregation and separation such as
\begin{align}
\label{Eqn:Aggregation2}
\partial_t \varrho=\partial_x \Bat{\varrho\,\partial_x \bat{g\at{\varrho}-a\ast\varrho}}\,,
\end{align}
where $g$ is now strictly increasing. Using $f:=g^{-1}$ and $U:=g\circ\varrho $ one easily shows that any solution to \eqref{eq:nlEv2} with $\si=1$ provides a steady state of  \eqref{Eqn:Aggregation2}. However, the set of all stationary equations is much larger as it has been shown in \cite{BFH14} using tailor-made fixed point arguments. In fact, assuming that $\varrho$ is compactly supported it suffices to fulfill a nonlinear equation on that support and this gives rise to more general steady states. We also refer to \cite{Kai17} for a related study on compactly supported minimizers of the corresponding energy functional.
\paragraph{More general models}  Many other application are also intimately related to nonlinear eigenvalue problems although the details and the underlying equations might be more involved. For instance, chains of coupled oscillators as described by the Kuramoto equation exhibit so-called chimera states, which can be characterized as solutions to a complex-valued variant of \eqref{eq:nlEv2}, see \cite{OMT08,Ome13,Ome18} for more details.  A second class of examples stems from the nonlocal analogues to the Allan-Cahn equation or \BMHR reaction-diffusion systems \EMHR  as the \BMHR corresponding \EMHR equations for steady states involve  both convolution operators and nonlinearities. However, since the latter are typically bistable one observes depinning effects, i.e. steady states exists only for certain parameters and but start to move at certain bifurcation values. \BMHR  This phenomenon \EMHR requires to study more complex equations which additionally involve continuous derivatives, see for instance \cite{BFRW97,AFSS16,FS15}.
\paragraph{ODE case}
We finally discuss a very special case, in which the eigenvalue problem \eqref{eq:nlEv} can be solved almost explicitly using ODEs. With the special choice
\begin{align}
\label{Eqn:ODECase}
a\at{x}=\tfrac{1}{2}\exp\bat{-\abs{x}}\,,\qquad \widehat{a}\at{k}=\frac{1}{1+k^2}
\end{align}
equation \eqref{eq:nlEv} is equivalent to the planar Hamiltonian ODE
\begin{align*}
\sigma \bat{U - U''}=f(U)\,.
\end{align*}
This can be seen by means of Fourier transform and inverting the linear differential operator on left hand side, and allows to study the existence of solutions using elementary phase plane analysis. In particular, for superlinear $f$ and any  $\si>f^\prime\at{0}$ there exists a unique homoclinic solution, whose orbit confines a family of periodic solutions. The exponentially decaying kernel \eqref{Eqn:ODECase} fits into our framework as it can be written as $a=b\ast b$, where $b$ is a modified Bessel function of second type and satisfies Assumption \ref{Ass:MainKernel} below although it exhibits a (logarithmic) singularity at the origin, see the third column in Figure \ref{fig:kernels}.  Notice that a similar class of nonlinear problems exists in higher space-dimensions with $a$ being the solution operator of a linear elliptic differential operator. There exists a vast literature on the corresponding nonlinear PDE but in this paper we focus on the one-dimensional setting with arbitrary kernels.
%
%
%
%----------------------------------------------------------------
\subsection{Main results}
%----------------------------------------------------------------
%
%
For the existence part of our work we rely on the following standing assumptions concerning the convolution kernel $b$ and \BMHR the nonlinearity $f$. A heuristic and numerical discussion of their necessity is postponed to \S\ref{sect:Approx}.\EMHR
\begin{assumption}[convolution  kernel]
\label{Ass:MainKernel}
The kernel function $b:\Rset \to \Rset$ satisfies
\begin{align}
 \label{ass:b1}
\int\limits_\Rset b(x) \dint x  =1\,,\qquad  \int\limits_\Rset  b(x) x^2 \dint x  < \infty
\,,\qquad  \int\limits_\Rset b^2(x)  \dint x  < \infty
\end{align}
and is supposed to be \BMHR nonnegative, \EMHR even, and unimodal as illustrated in Figure \ref{fig:kernels}.
\end{assumption}
\begin{figure}[ht!] %
\centering{%
\includegraphics[width=.95\textwidth]{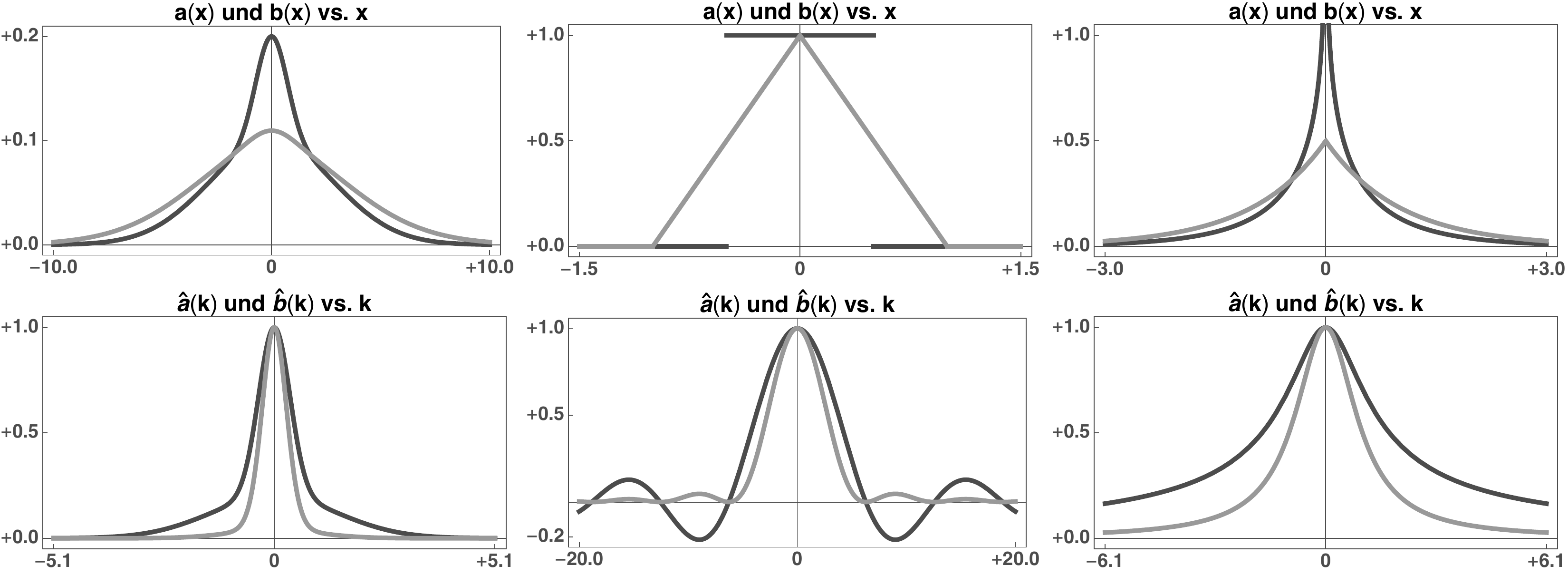}%
}%
\caption{Examples for the convolution kernels $a$ (light gray) and $b$ (dark gray) as well as their Fourier transforms, where \eqref{eqn:trafo} implies $\widehat{a}=\widehat{b}^2$. \emph{Left}. Generic case with smooth and rapidly decaying functions. \emph{Center}. Piecewise linear functions for FPUT chains. \emph{Right}. The ODE case \eqref{Eqn:ODECase}.
}%
\label{fig:kernels}%
\end{figure}%
\begin{assumption}[nonlinearity]
\label{Ass:MainNonl}
The function $f:\cointerval{0}{\infty}\to\Rset$
is twice continuously differentiable with
\begin{align}\label{ass:f1}
f(0)=0\,, \qquad \al:=f'(0) > 0\,,\qquad \beta:=f''(0) > 0
\end{align}
and is strictly superlinear due to $f^{\prime\prime}\at{r}>0$ for all $r> 0$.
\end{assumption}
In \S\ref{sect:setting} we first identify an underlying variational structure along with a constrained maximization problem for the eigenfunction $V$ that depends on a norm parameter $K$ and provides the eigenvalue \BMHR $\si>\alpha$ \EMHR as the corresponding Lagrange multiplier. A key ingredient to our approach \BMHR is the \EMHR  \emph{improvement operator} \BMHR in \eqref{eq:defT} below, \EMHR whose invariant cones enable us to impose shape constraints for $V$ without changing the Euler-Lagrange equation for maximizers. \S\ref{sect:Existence} ensures that the constrained optimization problem always admits  a solution, where the main technical challenge is to prove that the super\-linearity of $f$ favors localization of maximizing sequences and hence the existence of strongly convergent subsequences.
\par
Afterwards we discuss \BMHR in \S\ref{sect:Approx} \EMHR how to compute solutions to \eqref{eq:nlEv} by discretizing the improvement operator. The resulting numerical scheme works well, see Figure \ref{fig:waves} for some examples, although we are not able to prove its convergence due to the lack of uniqueness results. Finally, we characterize the decay of the eigenfunction at infinity in \S\ref{sect:decay} by splitting any eigenfunction into a compactly supported bulk part and remaining tail part, where the decay of the latter can be related to the properties of a modified kernel $a_c$.
\par
In \S\ref{sect:KdV} we study the limit of small eigenvalues $\si\gtrsim \alpha$ and show that the corresponding eigenfunctions converge after a suitable rescaling to the traveling wave profile of a Korteweg-de Vries (KdV) equation. Such results are well established for FPUT chains, see for instance \cite{FP99}, but our proof is more general and based on the constrained optimization problem.  \S\ref{sect:HE} is devoted to another asymptotic regime related to large eigenvalues and nonlinearities with algebraic singularity. For simplicity we restrict our considerations to the sample case
\begin{align}\label{eq:fsing}
  f(s)=\frac{1}{(1-s)^{m+1}} -1 \qquad\text{for some $m>1$ and all $0\leq s<1$}
\end{align}
and show for all sufficiently smooth kernels $b$ that the eigenfunction $V$ converges as $\si\to\infty$ to a well-defined multiple of $b$. Similar results have been shown for FPUT chains in \cite{FM02,Tre04,HM15,HM17} but the details are different due to the less regular kernel $b$, \BMHR see the discussion at the end of \S\ref{sect:scaling}. \EMHR
\par
Our main findings can informally be summarized as follows. The corresponding rigorous statements are given in Corollary \ref{Cor:Existence}, Corollary \ref{Corr:ExpDecay}, Theorem \ref{prop:kdv}, and Theorem \ref{thm:coarse}.
\begin{result}\quad
\begin{enumerate}
\item\emph{\ul{Variational existence and approximation\,:}} %
There exists a one-parameter family of solutions $\pair{\si}{V}$ to \eqref{eq:nlEv} with $\si>f^\prime\at{0}$ such that $V$ is \BMHR nonnegative, \EMHR even, and unimodal. This family can be constructed, both analytically and numerically, by a constrained optimization problem.
\item \emph{\ul{Decay at infinity\,:}} %
The (algebraic or exponential) decay properties of the eigenfunction $V$ \BMHR depend on the regularity \EMHR of the kernel $b$ as well as on the \BMHR value of \EMHR $\si$.
\item\emph{\ul{Asymptotics for small eigenvalues\,:}} %
If the eigenvalue $\si$ approaches $f^\prime\at{0}$ \BMHR from above, \EMHR then the eigenfunction $V$ converges
after a suitable rescaling to a traveling wave solution of a KdV equation.
\item\emph{\ul{Asymptotics for large eigenvalues\,:}} %
If $f$ exhibits an algebraic singularity and if $b$ is sufficiently smooth, then $V$ converges  as $\si\to\infty$ to a certain  multiple of the kernel $b$.
\end{enumerate}
\end{result}
The first two of these results are established under Assumptions \ref{Ass:MainKernel} and \ref{Ass:MainNonl} while our asymptotic analysis
for small and large eigenvalues requires a more restrictive setting, see \BMHR the refined \EMHR Assumptions \ref{Ass:KdV} and \ref{Ass:HE} below.
\begin{figure}[t!] %
\centering{%
\includegraphics[width=.95\textwidth]{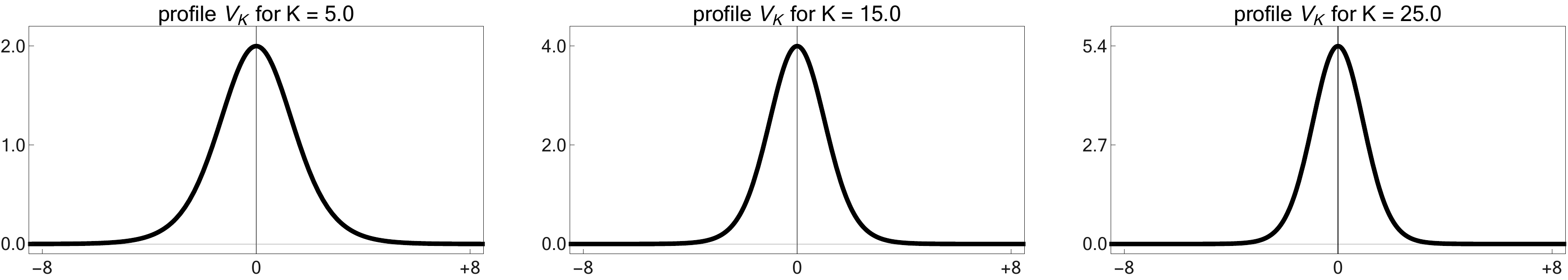}%
\\%
\includegraphics[width=.95\textwidth]{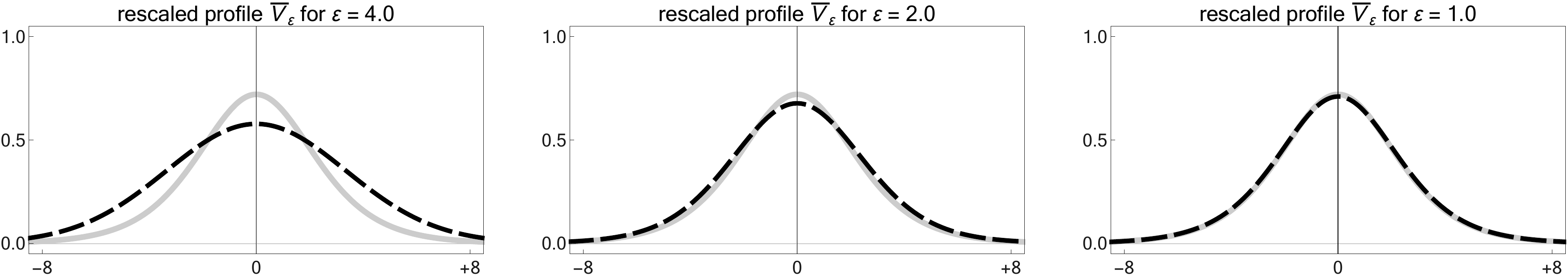}%
\\%
\includegraphics[width=.95\textwidth]{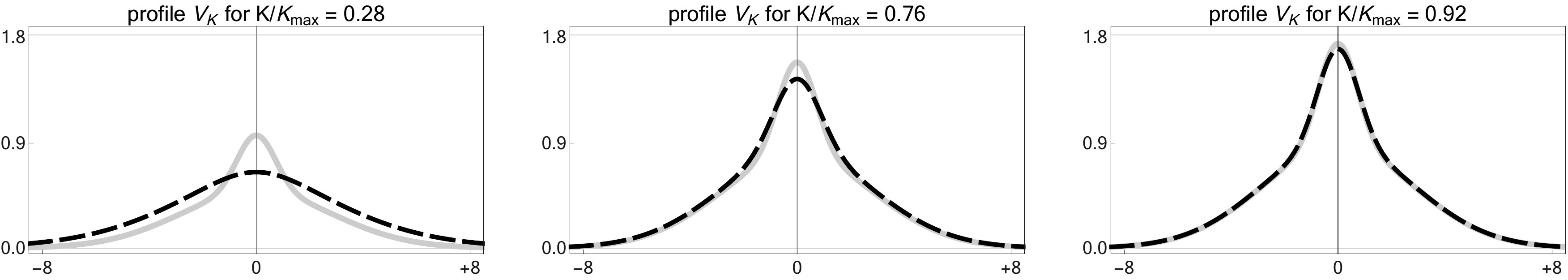}%
}%
\caption{Numerical simulations of solutions to \eqref{eq:nlEv}, computed with the scheme from \S\ref{sect:Approx}. \emph{Top}. Eigenfunction $V$ for several values of the norm constraint $K$ with $f\at{r}=\exp\at{r}-1$ and $b\at{x}=\exp\at{-x^2}/\pi$.  \emph{Center}. Scaled eigenfunctions in the KdV limit as discussed in \S\ref{sect:KdV} for $f$ and $b$ as in the \BMHR top row. \EMHR The gray curves represent the KdV profile defined in \eqref{prop:kdv.peqn1} and \eqref{prop:kdv.peqn2}. \emph{Bottom}. Convergence of eigenfunctions to a multiple of  $b$ (gray) as discussed in \S\ref{sect:HE} for the singular nonlinearity \eqref{eq:fsing}. $K_{\max}$ is given in \eqref{Eqn:KMax} and the computations are performed with $m=4$.}%
\label{fig:waves}%
\end{figure}%
\BMHR
\paragraph{Relation to earlier results} The main results are basically well-known in the special case of FPUT chains. For instance, the variational existence proof is an adaption of \cite{Her10} and based on preliminary work in \cite{FW94,FV99}. The exponential decay has been established in \cite{FP99} for small $\sigma$, in \cite[Lemma 4.1]{HR10} for some related waves, and generally in \cite{Pan19b}, but here we present a novel and shorter derivation. The asymptotic properties of small eigenvalues have been shown rigorously in \cite{FP99} using an implicit function argument, while here we provide an entirely variational proof. \cite{HM15} also studies the regime of large FPUT eigenvalues chains in a variational setting but
most of the results are not applicable here due to the different regularity properties of $b$ and $a$, see the more detailed discussion in \S\ref{sect:scaling}. A more basic result, however, still follows and resembles the findings in \cite{FM02}. 
\par
For more general convolution kernels, all main results are to our knowledge new contribution. In particular, we are not aware of any work concerning the equilibrium solutions to the biological aggregation model \eqref{Eqn:Aggregation1}. The existing results
for \eqref{Eqn:Aggregation2} in \cite{BFH14} concern slightly different problems with compactly supported functions and are derived by non-variational arguments.
\EMHR
%
%----------------------------------------------------------------
\subsection{Discussion and open problems}\label{sect:discussion}
%----------------------------------------------------------------
%
%
%
\paragraph{Other existence results}  The nonlinear eigenvalue problem \eqref{eq:nlEv2} can also be studied in a periodic setting, in which $V$ is supposed to be periodic with respect to $x$ while the convolutions kernels $a$ and $b$ are still defined on $\Rset$. As explained  below in \S\ref{sect:Approx}, our variational existence result can be generalized to this case provided that $f$ and $b$ comply with Assumptions \ref{Ass:MainKernel} and \ref{Ass:MainNonl}. It also possible to replace the right hand side in \eqref{eq:nlEv2} by a (finite or infinite) sum of nonlocal integral operators and this has already been done in \cite{HM19a} \BMHR for FPUT chains with nonlocal interactions, see also \cite{Pan19a} for a different approach. \EMHR A further candidate for generalizations are higher space dimensions, where the scalar functions \BMHR $V$ \EMHR and $b$ depend on a $d$-dimensional variable $x\in\Rset^d$, but the verification of the technical arguments might be more involved.
\par
At least for the periodic counterpart of \eqref{eq:nlEv2}, the existence of positive solutions can be shown by different methods. A first example are variants of the Crandall-Rabinowitz theory, see \cite{CR71,Rab71}, which constructs global solution branches that bifurcate from simple eigenvalues of the linearized eigenvalue equation $\la\, U= a\ast U $ with $\la=\si / f^\prime\at{0}$. A similar but local continuation approach has \BMHR been \EMHR applied in \cite{Ome13} to a complex-valued eigenvalue problem. \BMHR More \EMHR general small-amplitude results can be derived from spatial center manifold techniques \BMHR or direct bifurcation methods including Ljapunov-Schmidt reduction \EMHR as described in \cite{BS18,ST19}. \BMHR Furthermore, \EMHR since \eqref{eq:nlEv} admits a variational setting as explained in \S\ref{sect:setting}, one might also employ the Mountain Pass Theorem, see for instance \cite{Pan05} for related results on supersonic traveling waves in FPUT chains.
\par
Notice, however, that the restriction to periodic solutions simplifies the problems because the involved convolution operators turn out to be compact in this setting. One advantage of our approach is that it likewise applies to the solitary case, in which convolution operators are not compact. Moreover, it provides an existence result with unimodal and nonnegative eigenfunctions and gives rise to  an easy-to-implement approximation scheme, see \S\ref{sect:Approx} for more details.
\paragraph{Uniqueness and Stability}
Another open problem is the uniqueness of solutions to \eqref{eq:nlEv} or \eqref{eq:nlEv2} and we are not aware of any such result for generic nonlinear eigenvalues problems with superlinear nonlinearity. For FPUT chains, local uniqueness up to spatial shifts has been proven in certain asymptotic regimes which allow to tackle \eqref{eq:nlEv} by ODE arguments and perturbations techniques, see for instance the discussion in \cite{HM19b} and the references therein. These asymptotic results might be adapted to the scaling limits in \S\ref{sect:scaling}, but the general case remains open.  However, numerical simulations as discussed in \S\ref{sect:Approx} indicate the family of unimodal eigenfunctions that is provided by the constrained optimization approach in \S\ref{sect:Existence} is in fact unique, see Conjecture \ref{Conj:Uni} for a precise statement.
\par
We expect that the uniqueness of unimodal eigenfunctions can be linked to a nonlinear variant of the famous Krein-Rutman theory, but all known generalizations seem to be restricted to $1$-homogeneous nonlinearities, see for instance \cite{Mah07,Ara18}, and do not cover problems like \eqref{eq:nlEv}. Another key question concerns the linearized eigenvalue problem, in which the nonlinear superposition operator is replaced by the multiplication with the unimodal (and frozen) coefficient function $f^\prime\at{U}$. In this context it would be highly desirable to develop nonlocal variants of the classical Sturm-Liouville theory, but we are not aware of any relevant result \BMHR in this direction. \EMHR A closely related problems concerns the linear dynamical stability of solutions to \eqref{eq:nlEv} with respect to a given dynamical system (e.g., the FPUT lattice or the nonlocal PDE \eqref{Eqn:Aggregation1}).
\paragraph{Other nonlinearities}

Further open problems concern the existence, uniqueness and stability of solutions to vector-valued analogues to \eqref{eq:nlEv}, where a particular application are traveling waves in two-dimensional mechanical lattice systems as described in \cite{ChH18}. Finally, one might also study other types of nonlinearities such as sublinear or convex-concave  variants of $f$. In the context of FPUT chains there exists some partial results on the existence of solutions, see \cite{TV05, HR10,TV10,SZ12,HMSZ13}, but a complete theory is still missing. \BMHR We also refer to the numerical simulations at the end of \S\ref{sect:Approx}. \EMHR
%
%
%----------------------------------------------------------------
\section{Existence of solutions}\label{sect:Var} %
%----------------------------------------------------------------
%
%
For the subsequent analysis we write $F(r)=\int_{0}^{r} f(s) \dint s$ and notice that $F:\cointerval{0}{\infty}\to\Rset$ is strictly convex and superquadratic since Assumption \ref{Ass:MainNonl} ensures
\begin{align}
\label{eqn:superlin}
f\at{\la\,r }\geq \la f\at{r}, \qquad f^\prime\at{r}\,r \geq 2 F\at{r} \qquad \qquad \text{for all} \quad \la\geq1\quad\text{and}\quad r\geq0\,.
\end{align}
We further define the Fourier transform by
\begin{align*}
\widehat{U}\at{k}=\int\limits_\Rset\mhexp{-\iu kx}\,U\at{x}\dint{x}\,,\qquad
U\at{x}=\frac{1}{2\pi}\int\limits_\Rset\mhexp{+\iu kx}\,\widehat{U}\at{k}\dint{x}\,,
\end{align*}
and this implies the Plancherel identity
\begin{align}
\label{Eqn:Plancherel}
\bnorm{\widehat{U}}_2^2=2\pi\norm{U}_2^2
\end{align}
as well as
\begin{align*}
\widehat{U_1\ast U_2}=\widehat{U}_1\cdot\widehat{U}_2\,,\qquad \widehat{U}\at{0}=\int\limits_\Rset U\at{x}\dint{x}\,.
\end{align*}
Finally, introducing the functionals
\begin{align}
\notag%\label{eq:defmathcalP}
  \calP(V):= \int\limits_\Rset F\bat{\at{b \ast V}\at{x}} \dint x\,,\qquad
  \calK\at{V} :=  \frac12 \int\limits_\Rset\bat{ V\at{x}}^2\dint{x}\,,
\end{align}
the nonlinear eigenvalue problem  \eqref{eq:nlEv} can be written as
\begin{align}\label{eq:EUL}
  \sigma \partial \calK(V) = \partial \calP (V)\,,
\end{align}
where $\partial$ abbreviates the G\^ateaux differential operator. The latter formula is the starting point for our variational setting.
%
%
%----------------------------------------------------------------
\subsection{Variational setting}\label{sect:setting}
%----------------------------------------------------------------
%
%
In this paper we construct solutions to \eqref{eq:nlEv} by maximizing $\calP$ subject to
\begin{align*}
V\in \calC_K &:= \big\{W \in \calC \mid \calK(W)=K \big \} \,,
\end{align*}
which encodes both an $\fspaceL^2$-norm constraint and a shape constraint since the convex cone $\calC$ consists of all square integrable functions that are even, nonnegative, and unimodal. This reads
\begin{align*}
\calC&:= \overline{\big\{W \in \fspaceC^\infty_c(\Rset) \;\mid\;  W(x)=W(-x) \geq 0 \;\text{and}\;W'(x)=-W'(-x)\leq 0\;\text{for all}\;x \geq 0 \big\}}
\end{align*}
and implies that $\calC$ is closed under both weak and strong convergence in $\fspaceL^2\at\Rset$.
\par
The eigenvalue $\si$ in the Euler-Lagrange equation \eqref{eq:EUL} is clearly the Lagrange multiplier to the norm constraint. There is, however, no multiplier to the shape constraint. This might be surprising at a first glance but can be deduced from the properties of the \emph{improvement operator}
\begin{align}
\label{eq:defT}
  \calT(V) := \mu(V) \partial \calP(V)\,,\qquad
\mu(V) & := \frac{\|V\|_2}{\| \partial \calP(V)\|_2}
\end{align}
which is useful for both analytical and numerical issues.
\begin{lemma}[invariance properties]
\label{Lem:InvProps}
$V\in\calC$ implies $b\ast V\in\calC$ as well as  $f\at{b\ast V}\in \calC$.
\end{lemma}
\begin{proof}
Let $V\in\calC\subset\fspaceL^2\at\Rset$ be given. Standard arguments for convolution operators reveal that $U=b\ast V$ is square integrable as well as even and nonnegative. Under the additional assumption $V\in\fspaceC^1\at\Rset$ and fixing $x>0$ we compute
\begin{align*}
\frac{\dint}{\dint{x}}U\at{x}&=\int
\limits_{-\infty}^0 b\at{x-y}V^\prime\at{y}\dint{y}+
\int
\limits_{0}^{+\infty} b\at{x-y}V^\prime\at{y}\dint{y}
=\int\limits_{0}^{+\infty} V^\prime\at{y}\bat{b\at{x-y}-b\at{x+y}}\dint{y}\\&=
\int\limits_{0}^{x} \underbrace{V^\prime\at{y}}_{\leq0}\underbrace{\bat{b\at{x-y}-b\at{x+y}}}_{\geq0 }\dint{y}
+
\int\limits_{x}^{+\infty} \underbrace{V^\prime\at{y}}_{\leq0}\underbrace{\bat{b\at{y-x}-b\at{y+x}}}_{\geq0 }\dint{y}\leq0\,,
\end{align*}
where we used that both $b$ and $V$ are even as well as increasing and decreasing for negative and positive arguments, respectively. We have thus shown the unimodality of $U$ for smooth $V$ and the general case follows from approximation arguments. Finally, $f\at{U}\in\calC$ is a direct consequence of the continuity and the monotonicity of $f$.
\end{proof}
The next result implies  that the shape constraint does not in fact  alter the Euler-Lagrange equation \eqref{eq:EUL}.
\begin{proposition}[properties of the functionals and the operators]\label{prop:PT}
Let $K>0$. Then,  $\calP$, $\partial \calP$, and $\calT$ are well-defined on $\calC_K$. Moreover, the set $\calC_K$ is invariant under $\calT$ with
\begin{align}
\label{eq:PTP}
\calP\bat{\calT\at{V}} \geq \calP\at{V} \qquad \text{for all}\quad V \in \calC_K\,,
\end{align}
where equality holds if and only  if $V=\calT (V)$, i.e. if $V$ satisfies \eqref{eq:EUL} with $\sigma=(\mu(V))^{-1}$.
\end{proposition}
\begin{proof}
\emph{\ul{Properties of $\calP$ and $\partial\calP$}}\,:
Young's convolution inequality yields
\begin{align}
\label{eq:baV2}
  \|b \ast V\|_\infty  \leq \|b\|_2 \|V\|_2 = \|b\|_2  \sqrt{2K}\,,\qquad
 \|b \ast V\|_2  \leq \|b\|_1 \|V\|_2\leq  \|b\|_1 \sqrt{2K}\,,
\end{align}
so the properties of $f$ -- see Assumption \ref{Ass:MainKernel} and \eqref{eqn:superlin} ---  imply the pointwise estimate
\begin{align*}
0 \leq F(b \ast V) \leq C \bat{b \ast V}^2 \qquad \text{with}\qquad
C:= \sup_{0 \leq r \leq \|b \ast V\|_\infty}\frac{f(r)}{r}\leq \frac{f(\sqrt{2K})}{\sqrt{2K}}   \,.
\end{align*}
Integration with respect to $x$ reveals that $\calP$ is well-defined on $\fspaceL^2\at\Rset$ and satisfies
\begin{align}\nonumber
  0 &\leq \calP(V) \leq C\|b \ast V\|_2^2  \leq C \| V\|_2^2
= f(\sqrt{2K})\sqrt{2K}
\end{align}
for any $V$ with $\norm{V}_2^2=2K$. Similarly we have
\begin{align*}
0 \leq f(b \ast V) \leq f'\bat{\sqrt{2K}}(b \ast V),
\end{align*}
which implies that $\partial \calP\at{V}$ is also well-defined with
\begin{align*}
\babs{\skp{\partial \calP(V)}{ W}} \leq \int\limits_\Rset f\at{b \ast V} \,\at{b \ast W }\dint{x}\leq f'(\sqrt{2K}) \|b \ast V\|_2 \|b \ast W\|_2\leq  f'(\sqrt{2K}) \sqrt{2K} \| W\|_2
\end{align*}
for any $W \in \fspaceL^2\at\Rset$.
\par
\emph{\ul{Improvement operator}}\.: %
Both $F$ and $\calP$ are convex due to $F''(r)=f'(r)\geq 0$ and this implies
\begin{align}\label{eq:convP}
\calP(W)-  \calP(V) \geq \bskp{ \partial \calP(V)}{ W-V }
\end{align}
for all $V,W \in\fspaceL^2\at{\Rset}$. Since $\calP\at{V}>0=\calP\at{0}$ holds for any $V\in C_K$, we deduce that $\partial\calP\at{V}\neq0$ because otherwise the evaluation of \eqref{eq:convP} with $W=0$ would lead to a contradiction. We therefore have
\begin{align}
\label{eq:muV}
 \mu(V)\neq0\,,
\end{align}
\BMHR so the \EMHR operator $\calT$ is well-defined on $\calC_K$.  Moreover, it maps $\calC_K$ into itself since $\norm{\calT\at{V}}_2=\norm{V}_2$ holds by construction and because both the convolution with $b$ and the superposition with $f$ respects the unimodality, evenness and nonnegativity of functions, see Lemma \ref{Lem:InvProps}. Finally, setting $W= \calT\at{V}$ in \eqref{eq:convP} and using \eqref{eq:muV} as well as $\nnorm{\calT\at{V}}_2=\norm{V}_2$ we get
\begin{align*}
  \calP\bat{\calT\at{V}}- \calP\at{V}   \geq\bskp{\partial \calP(V)}{\calT(V)-V } = \frac{ \bskp{ \calT(V)}{ \calT(V)-V }}{\mu\at{V}}= \frac{\| \calT(V) -V \|_2^2}{2 \mu(V)},
\end{align*}
and conclude that $\calP\bat{\calT\at{V}}=\calP\at{V}$ holds if and only $V$ is a fixed point of $\calT$.
\end{proof}
In the next section we prove that $\calP$ attains its maximum on $\calC_K$.
%
%
%----------------------------------------------------------------
\subsection{Existence of constrained maximizers}
\label{sect:Existence}
%----------------------------------------------------------------
%
%
%
Our goal is to show that any maximizing sequence for $\calP$ in $\calC_K$ admits strongly convergent subsequences so that the existence of solutions to the constrained optimization problem
can be deduced by the Direct Method. In a first step we study
\begin{align}
\label{Eqn:MaximaOfEnergy}
  P\at{K}  := \sup_{V \in \calC_K} \calP(V)\,,\qquad\qquad
  Q\at{K} := \sup_{V \in \calC_K} \calQ(V)
\end{align}
with
\begin{align*}
\calQ(V)= \frac{f'(0)}{2} \int\limits_{\Rset} \bat{(b \ast V)\at{x}}^2\,,
\end{align*}
i.e., we compare the nonlinear functional $\calP$ with its quadratic counterpart \BMHR $\calQ$. \EMHR
\begin{lemma}[key estimate for maxima]\label{Lem:SuperEstimate}%
For any $K>0$, the variational problem is super-quadratic in the sense that
\begin{align}
\label{Lem:SuperEstimate.Eqn1}
P(K)  > Q(K)={f^\prime\at{0}\,K}
\end{align}
holds for the quantities \BMHR defined \EMHR in \eqref{Eqn:MaximaOfEnergy}.
\end{lemma}
\begin{proof}
\emph{\ul{Contributions from quadratic terms and test functions}}\,: %
We first observe that Young's inequality $\nnorm{b\ast V}_2\leq \nnorm{b}_1\nnorm{V}_2$ and the normalization condition $\nnorm{b}_1=1$ imply
\begin{align}
\label{Lem:SuperEstimate.PEqn1}
Q(K) = \sup_{V \in \calC_K} \frac{f'(0)}{2} \bnorm{b\ast{V}}_2^2\leq
\sup_{V \in \calC_K} \frac{f'(0)}{2} \bnorm{{V}}_2^2\leq f'(0)\, K\,.
\end{align}
For any $L>1$ we now consider the function $V_L\in \calC_K$ defined by
\begin{align*}
V_L\at{x} =\sqrt{\frac{K}{L}} \chi_{[-L,L]}(x)\,.
\end{align*}
Since $U_L:=b \ast V_L$ is even and unimodal according to Lemma \ref{Lem:InvProps} we get
\begin{align*}
U_L\at{L-L^{1/3}}\geq \sqrt{\frac{K}{L}} \int\limits^{+L^{1/3}} _{\BMHR-L^{1/3}\EMHR}b\at{x}\dint{x}\geq \sqrt{\frac{K}{L}}\at{1-2\int\limits_{L^{1/3}}^\infty b\at{x}\dint{x}}\geq  \sqrt{\frac{K}{L}}\at{1-\frac{C}{L^{2/3}}}\,,
\end{align*}
where we also used the moment estimates for $b$ from \eqref{ass:b1}, and this implies
\begin{align*}
\calQ\at{V_L}&\geq \frac{f'(0)}{2} \int\limits_{-L+L^{1/3}}^{L- L^{1/3}} U_L^2\at{x}\dint{x}\geq f'(0) \bat{L-L^{1/3}} U^2_L\at{L-L^{1/3}} \geq f^\prime\at{0}K-\frac{C}{L^{2/3}}\,.
\end{align*}
In particular, we have $Q\at{K}\geq  \lim_{L\to\infty} \calQ\bat{V_L}=f^\prime\at{0}\,K$ and obtain in combination with \eqref{Lem:SuperEstimate.PEqn1} the formula for $Q\at{K}$.
\par
\emph{\ul{Contributions from cubic terms}}\,:
On the other hand, thanks to $\nnorm{U_L}_\infty\leq \nnorm{b}_1\nnorm{V_L}_\infty\leq C L^{-1/2}$ and the regularity of $f$ combined with $F^{\prime\prime\prime}\at{0}=f^{\prime\prime}\at{0}>0$ we estimate
\begin{align*}
  \calP(V_L) -\calQ(V_L)  &= \int\limits_{-\infty}^{+\infty} \at{F\bat{U_L\at{x}} -\frac{f'(0)}{2}  U_L^2\at{x}} \dint{x }\geq
  c  \int\limits^{L-L^{1/3}}_{-L+L^{1/3}}
  U_L^3\at{x}\dint{x}\\&\geq c\,\at{L-L^{1/3}} \frac{K^{3/2}}{L^{3/2}}\at{1-\frac{C}{L^{2/3}}}^{\BMHR 3 \EMHR} \geq \frac{c}{L^{1/2}}
   \end{align*}
for \BMHR all sufficiently large $L$. \EMHR Combining all partial results we find
\begin{align*}
P\at{K}\geq \calP\at{V_L}=\BMHR \calQ\at{V_L}+\calP\at{V_L} - \calQ\at{V_L} \geq Q\at{K}-\frac{C}{L^{2/3}}+\frac{c}{L^{1/2}}\EMHR
\end{align*}
and the thesis follows from choosing $L$ finite but large enough.
\end{proof}
The superquadraticity relation \eqref{Lem:SuperEstimate.Eqn1} implies the concentration compactness of maximizing sequences within the cone $\calC$. The analogous conditions for traveling waves in atomic and peridynamical systems have been introduced and studied in \cite{Her10,HM19a}.
\begin{proposition}[strong compactness of maximizing sequences]\label{lem:ex1}
Any sequence $(V_n)_{n \in \BMHR \Nset\EMHR}$ in $\calC_K$ with $\calP(V_n) \to P(K)$ admits a subsequence that converges strongly in $\fspaceL^2\at\Rset$.
\end{proposition}
\begin{proof} \emph{\ul{Preliminaries}}\,:
There exists a (not relabeled) subsequence such that
\begin{align}
\label{Eqn:WeakConv}
V_n \quad\xrightarrow{\;n\to\infty\;}\quad V_\infty \quad\text{weakly in }\quad \fspaceL^2(\Rset)
\end{align}
for some $V_\infty $ which belongs to the cone $\calC$ (which is convex and closed) and satisfies $\|V_\infty\|_2^2 \leq 2K$. We aim to show
\begin{align}\label{eq:normconv}
  \|V_\infty\|_2^2 \geq 2K,
\end{align}
because this \BMHR estimate \EMHR implies in combination with the weak convergence the desired strong convergence in the Hilbert space $\fspaceL^2(\Rset)$. To this end we fix $\eps>0$ arbitrarily and consider for a given cut-off parameter $0<M<\infty$ the modified functionals
\begin{align*}
\widetilde{\calP}\at{V}:=\int\limits_\Rset F\bat{\widetilde{b}\ast\BMHR {V}\EMHR}\dint{x}\,,\qquad 
\widetilde{\calQ}\at{V}:=\frac{f^\prime\at{0}}{2}\int\limits_\Rset \bat{\widetilde{b}\ast\BMHR {V}\EMHR}^2\dint{x}
\end{align*}
with
\begin{align*}
\tilde{b}\at{x} & := b\at{x}\chi_{[-M,M]}\at{x}
\end{align*}
and $\chi_I$ denoting the indicator function of the interval $I$. We further split (for both finite and infinite $n$) the maximizing sequence \BMHR and its weak limit \EMHR according to
\begin{align}
\label{eqn:CutOff}
   \widetilde{V}_n\at{x}:= V_n \at{x}\chi_{\ccinterval{-M^2}{M^2}}\at{x}\,,\qquad
    \ol{V}_n\at{x}  =V_n\at{x}- \widetilde{V}_n\at{x}
\end{align}
and observe that
\begin{align}\label{eq:pyth}
   \|\tilde{V}_n\|_2^2+  \|\bar{V}_n\|_2^2= \|{V}_n\|_2^2= 2K
\end{align}
holds by construction for $n\neq\infty$. \BMHR Notice that $\tilde{b}$ and $\tilde{V}$ are defined by different cut offs.\EMHR
\par
\emph{\ul{Approximation formulas}}\,: %
Using the positivity of $b$ and $V_n$, the properties of $f$, and estimates as in \eqref{eq:baV2} we establish the Lipschitz estimate
\begin{align*}
0\leq \calP\at{V_n}-\widetilde{\calP}\at{V_n}&\leq C \int\limits_\Rset \babs{\bat{b\ast V_n}\at{x} }\,\babs{\bat{b\ast V_n}\at{x} -\bat{\widetilde{b}\ast V_n}\at{x}} \dint{x}
\leq C\bnorm{b-\tilde{b}}_2\leq \eps
\end{align*}
provided that $M$ is chosen sufficiently large. We further know that $\tilde{b}\ast\widetilde{V}_n$ and $\tilde{b}\ast\ol{V}_n$ are supported in $\abs{x}\leq M^2+M$ and $\abs{x}\geq M^2-M$, respectively, and in combination with the pointwise estimates
\begin{align*}
0\leq \widetilde{b}\ast \widetilde{V}_n\leq \widetilde{b}\ast V_n\,,\qquad 0\leq \widetilde{b}\ast \ol{V}_n\leq \widetilde{b}\ast V_n
\end{align*}
we obtain
\begin{align*}
\abs{\widetilde{P}\bat{V_n}-
\widetilde{P}\bat{\widetilde{V}_n}-\widetilde{P}\bat{\ol{V}_n}}
\leq C\int\limits_{M^2-M}^{M^2+M} F\bat{\widetilde{b}\ast V_n }\dint{x}\,.
\end{align*}
Moreover, $\tilde{b}\ast V_n\in\calC$  implies the uniform tightness estimate
\begin{align*}
0\leq \bat{\widetilde{b}\ast V_n}\at{x}\leq
 \frac{\bnorm{\widetilde{b}\ast V_n}_2}{\sqrt{2\abs{x}}}
\leq  \frac{\bnorm{\widetilde{b}}_1\bnorm{V_n}_2}{\sqrt{2\abs{x}}}\leq\frac{C}{\sqrt{x}} \,,
\end{align*}
and choosing $M$ sufficiently large we find
\begin{align*}
C\int\limits_{M^2-M}^{M^2+M} F\bat{\widetilde{b}\ast V_n }\dint{x}\leq C \int\limits_{M^2-M}^{M^2+M}\frac{\dint{x}}{x}\leq C\ln\at{\frac{M^2+M}{M^2-M}}\leq \eps\,.
\end{align*}
By a similar argument we derive
\begin{align*}
\babs{\widetilde{P}\at{
\ol{V}_n}-\widetilde{Q}\at{
\ol{V}_n}}\leq \eps
\end{align*}
from Taylor expanding $F$ around $0$ and obtain in summary the estimate
\begin{align}
\label{Eqn:ToolEst1}
\abs{\calP\bat{V_n}-\widetilde{\calP}\bat{\widetilde{V}_n}-\widetilde{\calQ}\bat{\ol{V}_n} }\leq 3\,\eps
\end{align}
along the chosen subsequence as well as for $n=\infty$. Finally,  the estimate
\begin{align}
\label{Eqn:ToolEst2}
\babs{\widetilde{Q}\bat{\ol{V}_\infty}}\leq \eps
\end{align}
can be guaranteed by enlarging $M$ if necessary.
\par%
\emph{\ul{Scaling argument and limit}}\;:
The weak convergence \eqref{Eqn:WeakConv} implies that
\begin{align*}
\widetilde{b}\ast \widetilde{V}_n
 \quad\xrightarrow{\;n\to\infty\;}\quad\widetilde{b}\ast \widetilde{V}_\infty
\end{align*}
holds pointwise, but since  all functions $\widetilde{b}\ast \widetilde{V}_n $ are supported in $\ccinterval{-M^2-M}{M^2+M}$ as well as uniformly bounded by $\bnorm{\widetilde{b}}_2\sqrt{2K}$ this convergence also holds strongly in \BMHR $\fspaceL^2\at{\Rset}$. \EMHR For all sufficiently large $n$, we therefore have
\begin{align*}
\widetilde{P}\bat{\widetilde{V}_n}\leq \widetilde{P}\bat{\widetilde{V}_\infty}+\eps
\end{align*}
and by construction we can further assume that
\begin{align*}
\calP\at{V_n}\geq P\at{K}-\eps\,.
\end{align*}
Combining the last two estimates with \eqref{Eqn:ToolEst1} (evaluated for both finite and infinite $n$) as well as \eqref{Eqn:ToolEst2} we show that
\begin{align*}
P\at{K}\leq \calP\bat{\widetilde{V}_\infty}+\calQ\bat{\ol{V}_n}+C\,\eps\,
\end{align*}
holds for $n$ large enough. \BMHR Since $\calP$ is superquadratic,  we further observe\EMHR
\begin{align*}
\calP(\widetilde{V}_\infty) = \frac{\|\widetilde{V}_\infty\|^2_2}{2K}\calP\left( \frac{\sqrt{2K}}{\|\widetilde{V}_\infty\|^2_2} \widetilde{V}_\infty  \right) \leq \frac{\|\widetilde{V}_\infty\|^2_2}{2K} P(K),
\end{align*}
while the homogeneity of $\calQ$ guarantees
\begin{align*}
\calQ(\ol{V}_n) \leq  \frac{\|\ol{V}_n\|^2_2}{2K} Q(K)\,.
\end{align*}
This gives
\begin{align}
\label{eq:PPestY}
P(K) \leq  \frac{\|\widetilde{V}_\infty\|^2_2}{2K} P(K) + \frac{\|\ol{V}_n\|^2_2}{2K} Q(K) + C\, \eps
\end{align}
and writing  $Q(K)= P(K)- (P(K)-Q(K))$ yields
\begin{align}
\label{eq:PPestX}
  P(K) +  \frac{\|\ol{V}_n\|^2_2}{2K} \bat{P(K)-Q(K)}\leq \frac{\|\widetilde{V}_\infty\|^2_2+\|\ol{V}_n\|^2_2}{2K} P(K) + C\,\eps\,.
\end{align}
We further have
\begin{align*}
\frac{\|\widetilde{V}_\infty\|^2_2+\|\ol{V}_n\|^2_2}{2K} \leq 1 +  \eps
\end{align*}
for all sufficiently large $n$ thanks to \eqref{eq:pyth} and because
\eqref{Eqn:WeakConv}+\eqref{eqn:CutOff} imply  the weak convergence $\widetilde{V}_n\to\widetilde{V}_\infty$ and hence $\nnorm{\widetilde{V}_\infty}_2\leq \liminf_{n\to\infty}\nnorm{\widetilde{V}_n}_2$. The right hand side in \eqref{eq:PPestX} can thus be  estimated from above by $P\at{K}+C\eps$, so rearranging terms and using Lemma \ref{Lem:SuperEstimate} we obtain
\begin{align*}
\|\ol{V}_n \|_2^2 \leq \frac{2\,K\,C\, \eps}{P(K)-Q(K)} \leq C\,\eps
\end{align*}
for all large $n$. Inserting this into \eqref{eq:PPestY} we get
\begin{align*}
 P(K) \leq \frac{\|\widetilde{V}_\infty\|^2_2}{2K} P(K) + C\,\eps\,,
\end{align*}
where the constant $C$ does not depend on $\eps$. \BMHR Finally, since \EMHR $\eps>0$ was arbitrary, \BMHR we verify \EMHR \eqref{eq:normconv} \BMHR thanks to $\|\widetilde{V}_\infty\|^2_2\leq \|V_\infty\|^2_2$ and $P\at{K}>0$. \EMHR
\end{proof}
\begin{corollary}[existence of solutions]\label{Cor:Existence}
For any $K>0$ there exists a solution $\pair{\si}{V}$ to \eqref{eq:nlEv} with
\begin{align}
\label{Cor:Existence.Eqn1}
\si \geq K^{-1} \BMHR P\at{K} >f^\prime\at{0}\,,\qquad \calP\at{V}=P\at{K}\EMHR
\end{align}
as well as  $V\in\calC_K$.
\end{corollary}
\begin{proof}
According to Proposition \ref{lem:ex1}, there exists a maximizer $V$ of $\calP$ in $\calC_K$ which can be constructed  as an accumulation point of a maximizing sequence. For any maximizer we have $\calP\bat{\calT\at{V}}=\calP\bat{V}$, so Lemma \eqref{prop:PT} ensures the validity of the Euler-Lagrange equation \eqref{eq:EUL} with multiplier $\si=1/\mu\at{V}$. Finally, testing \eqref{eq:EUL} with $V$ and using the symmetry of $b$ as well as \eqref{eqn:superlin} we get
\begin{align*}
2\,\si\,K = \bskp{f\bat{b\ast V}}{b\ast V}\geq 2\,\calP\at{V}=2\,P\at{K}
\end{align*}
and obtain in combination with Lemma \ref{Lem:SuperEstimate}  the lower bound for $\si$.
\end{proof}
We emphasize that we have no uniqueness result for the solutions provided by Corollary \ref{Cor:Existence}, neither for the maximizer $V$ nor the multiplier $\si$. We  also cannot exclude the existence of further solutions corresponding to saddle points of the functional $\calP$ restricted to $\calC_K$. However, numerical simulations as discussed below indicate that there exists a unique maximizer for a huge class of convolution kernels and superlinear nonlinearities.
%
%
%
%----------------------------------------------------------------
\subsection{Periodic solutions and numerical computation}
\label{sect:Approx}
%----------------------------------------------------------------
%
%
The variational existence proof can be generalized to periodic waves. In fact, one easily introduces the analogues of $\calP$, $\calK$, $\calC_K$, and $\calT$  in the space of all functions that are square integrable on the perodicity cell $\ocinterval{-L}{+L}$ and the results in Lemma \ref{Lem:InvProps} and Proposition \ref{prop:PT} can be proven along the same lines. The compactness argument in the proof of Proposition \ref{lem:ex1}  even simplifies since convolution operators are compact in a periodic setting and map weakly convergent sequences (which always exists due to the norm constraint) into strongly convergent one. For small values of $L$ and $K$ it might happen that the unimodal maximizer is a constant function but if $L$ or $K$ are sufficiently large, the strict superquadraticity of $F$ favors the localization of maximizers. Moreover, periodic maximizers converge as $L\to\infty$ to solitary solutions of \eqref{eq:nlEv}.  We refer to \cite{HM19a} for a similar discussion in the context of  atomic chains and to \cite{Wei99} for the general phenomenon of localization thresholds.
\bigpar
The improvement operator \eqref{eq:defT} can be iterated in the following approximation scheme with parameter $K>0$:
\begin{align}
\label{Eqn:Scheme}
\text{Guess $U_0\in\calC_k$ and compute $U_j$ recursively via $U_{j}=\calT\at{U_{j-1}}$ for $j\in\Nset$.}
\end{align}
The estimate \eqref{eq:PTP} ensures that $\calP$ increases along the resulting sequence $\at{U_{j}}_{j\in\Nset}\subset\calC_K$ and hence that $\calP\at{U_j}$ converges as $j\to\infty$ to a well-defined limit. Exploiting the arguments in the proofs of Propositions \ref{prop:PT} and \ref{lem:ex1} we can also show that any accumulation point must be  a solution to the nonlinear eigenvalue problem \eqref{eq:nlEv},  but due to the lack of uniqueness results we are not able to conclude that accumulation points are unique and independent of $U_0$. We further mention that variants of the improvement dynamics have been introduced in \cite{FV99,EP05,Her10}.
\par
\begin{figure}[t!] %
\centering{%
\includegraphics[width=.95\textwidth]{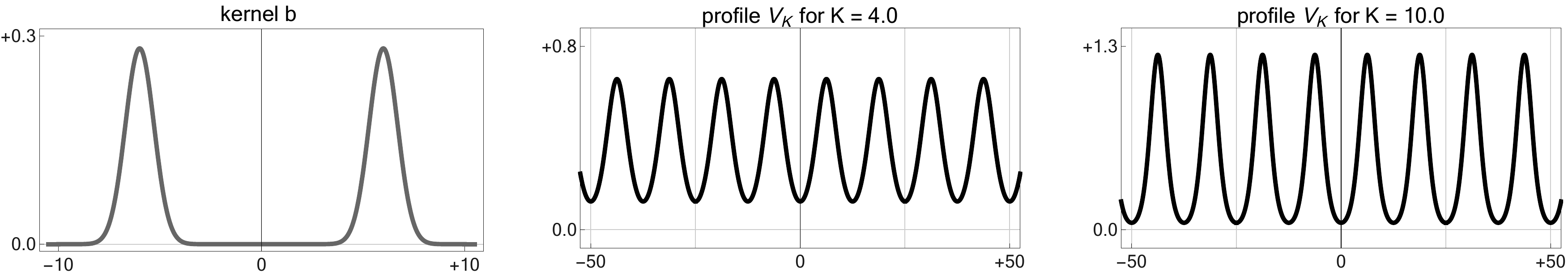}%
\\\smallskip
\includegraphics[width=.95\textwidth]{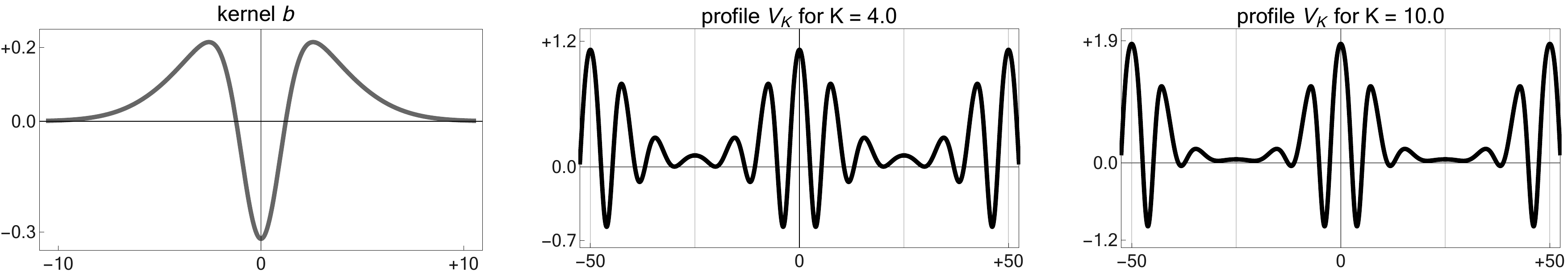}%
}%
\caption{\BMHR 
Numerical solutions with superlinear function $f\at{r}=\exp\at{r}$ and periodicity length $2L=50$ for two nonunimodal convolution kernels $b$. \emph{Top}. The kernel is still nonnegative but the numerical eigenfunctions produced by the improvement dynamics are neither localized nor unimodal with respect to the prescribed periodicity cell. \emph{Bottom}.  For a sign-changing kernel, the numerical eigenfunctions are
no longer nonnegative and exhibit a rather complicated shape. \EMHR}%
\label{fig:other2}%
\end{figure}%
\BMHR The time-discrete \EMHR improvement dynamics in \eqref{Eqn:Scheme} is also useful for computational issues and can easily be complemented by a spatial discretization:
\begin{enumerate}
\item
Choose a large length parameters $L<\infty$ as well as a small discretization parameter $\eps>0$ and replace functions $V\in\fspaceL^2\at\Rset$ by discrete and $2L$-periodic functions on the lattice $\eps\Zset$.
\item
Approximate all integrals in the definition of the improvement operator $\calT$ by Riemann sums with respect to $x\in\ocinterval{-L}{L}\cap\eps\Zset$
\end{enumerate}
The resulting numerical scheme exhibits good and robust convergence properties in practice and was used to produce the data presented in  Figure \ref{fig:waves}. In particular, numerical simulations performed with different choices of $b$ and  strictly superlinear $f$ indicate the validity of the following hypothesis.
\begin{conjecture}[uniqueness of unimodal maximizers]
\label{Conj:Uni}
For any kernel $b$ as in Assumption \ref{Ass:MainKernel} and any nonlinearity as in Assumption \ref{Ass:MainNonl} there exists  a unique maximizer of $\calP$ in $\calC_K$, which is moreover a global attractor for the improvement dynamics \eqref{Eqn:Scheme}.%
\end{conjecture} %
\begin{figure}[t!] %
\centering{%
\includegraphics[width=.95\textwidth]{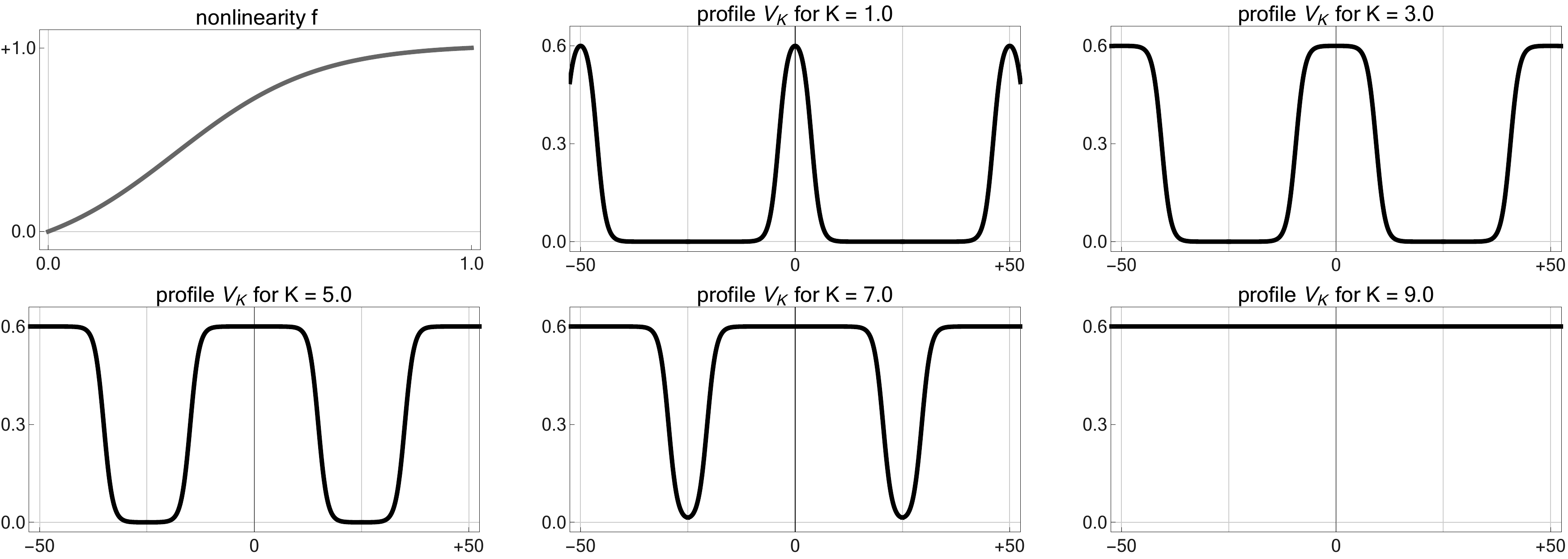}%
}%
\caption{\BMHR Periodic maximizers for a convex-concave nonlinearity $f$ with $2L=50$ and Gaussian kernel $b\at{x}=\exp\at{-x^2}/\pi$. \emph{Interpretation.}   $\at{i} $ For any $L<\infty$,
we expect to find a delocalization threshold $K_L<\infty$. $\at{ii}$ There still exists the family $\pair{\si_K}{V_K}_{K>0}$ of solitary solutions to \eqref{eq:nlEv} but the unimodal eigenfunction $V_K\in\calC_K$ exhibits for large $K$ a huge plateau whose height is basically independent of $K$. \EMHR}%
\label{fig:other1}%
\end{figure}%
\paragraph{Other classes of kernels and nonlinearities}
\BMHR The discretized improvement dynamics can also be started for convolution kernels or nonlinearities that do not meet the requierements in Assumptions \ref{Ass:MainKernel} or \ref{Ass:MainNonl}, respectively. For instance, Figure \ref{fig:other2} shows numerically computed eigenfunctions for two nonunimodal kernels $b$ but unimodal initial data. In the first example, $b$ is the sum of two localized but nonnegative peaks which are sufficiently narrow and separated. Our simulations suggest that the periodic variant of scheme \eqref{Eqn:Scheme} converges to a limit profile but the corresponding eigenfunction has a smaller periodicity length. In particular, we expect that there still exist local and global maximizers of the functional $\calP$ subject to the $\fspaceL^2$-norm constraint in the periodic setting, but there is probably no analogue to Conjecture \ref{Conj:Uni}. Moreover, the lack of localization indicates that solitary eigenfunctions might not exist in this case. The second example in Figure \ref{fig:other2} concerns  convolution kernels with real roots. In this case, the numerically computed eigenfunctions are no longer unimodal but still localized with oscillatory tails and this hints at the existence of solitary eigenfunctions with more general shape. In summary, it remains a challenging task to study the eigenvalue problem \eqref{eq:nlEv} with superlinear functions $f$ but more general kernels $b$.
\par
The simulations in Figure \ref{fig:other1} are performed with a unimodal kernel but a convex-concave function $f$ which switches from superlinear to sublinear growth. For small values of $K$, we still find unimodal and strongly localized eigenfunctions $V_K\in\calC_K$ as predicted in \S\ref{sect:Existence}. For larger values of $K$, however, the concave branch of $f$ implies strong delocalization effects. In the periodic setting $L<\infty$, the numerical data indicate that the maximizer of $\calP$ in $\calC_K$ is constant for large values of $K$. The consistent expectation for the limit $K\to\infty$ in the solitary case $L=\infty$ are unimodal eigenfunctions with huge plateaus and rapidly decaying transition layers, where the height of the plateau is asymptotically constant so that the plateau width scales with $\sqrt{K}$. Finally, for monotone but globally sublinear or even non-monotone functions $f$ we expect a much stronger impact of the aforementioned delocalisation effects and the existence of localised eigenfunctions is basically open, although there exists some preliminary results for  FPUT chains, see the references at the end of \S\ref{sect:discussion}.
\EMHR

%
%
%
%
%----------------------------------------------------------------
\subsection{Decay estimates}
\label{sect:decay} %
%----------------------------------------------------------------
%
%
In this section we characterize the spatial decay of solutions $\pair{\si}{U}$ to the nonlinear eigenvalue problem \eqref{eq:nlEv2}. Our result does not rely on variational arguments, so it applies to the family from Corollary \ref{Cor:Existence} but also to solutions provided by any other method. A similar result has recently been derived in \BMHR \cite{Pan19b} \EMHR for atomic chains with nonlocal interactions. The corresponding proof relies on (abstract) spectral theory and can be adapted to the nonlinear eigenvalue problem \eqref{eq:nlEv} provided that the kernel $b$ decays exponentially fast. Our method is both more elementary and more general but exploits similar ideas and concepts. In particular, both approaches require the nonlinear eigenvalue $\si$ to be larger than the essential spectrum of the linearized operator, and this property follows in our context from Assumption \ref{Ass:MainKernel} and the lower bound for $\si$ in \eqref{Cor:Existence.Eqn1}.
\par
We first study a linear auxiliary operator and \BMHR investigate \EMHR  the nonlinear problem \eqref{eq:nlEv2} afterwards.
\begin{lemma}[auxiliary result]\label{lem:decay1}
For any $0 < c< 1=\widehat{b}\at{0}$, the equation
\begin{align}
\notag%\label{eq:decaylem1}
W- c \,b \ast b \ast W = b \ast b \ast G
\end{align}
defines an linear and bounded operator $G \mapsto W=: A_c \,G$, which maps  $\fspaceL^2(\Rset)$ into itself. Moreover, this operator can be written as
\begin{align}
\label{lem:decay1.eqn1}
 A_c\, G= a_c\ast G\,,\qquad   a_c :=\sum_{m=0}^{\infty} c^m \underbrace{(b\ast b) \ast \ldots \ast (b\ast b)}_{\text{$m+1$ times }}
\end{align}
and preserves the unimodality, \BMHR nonnegativeness, \EMHR and evenness of functions.
\end{lemma}
\begin{proof}
Assumption \ref{Ass:MainKernel} combined with the Young estimates
\begin{align}
\label{lem:decay1.eqn2}
 \nnorm{b \ast b}_1= \nnorm{b}_1 ^2=1\,,\qquad \nnorm{\hat{b}}_\infty\leq \norm{b}_1=1
\end{align}
reveals that the function
\begin{align}\label{eq:decaylemFour}
\widehat{a}_c\at{k}:= \frac{\widehat{b}^2\at{k}}{1 - c \,\widehat{b}^2\at{k}}=\BMHR\widehat{b}^2\at{k}\sum_{m=0}^\infty \bat{c\, \widehat{b}^2\at{k}}^m\EMHR
\end{align}
is nonnegative and bounded. \BMHR Using Fourier transform we thus deduce \EMHR that the operator $A_c$ is \BMHR a well-defined \EMHR pseudo-differential operator with symbol function $\widehat{a}_c$. Formula \eqref{lem:decay1.eqn1} is just the Neumann representation of $A_c$ and implies  the claimed preservation properties thanks to Lemma \ref{Lem:InvProps}.
\end{proof}
\begin{corollary}
[general decay estimate for $U$]
\label{Corr:ExpDecay}%
Let $\pair{\si}{U}$ be a solution to \eqref{eq:nlEv2} and $c$ be a fixed constant with $\si^{-1}f^\prime\at{0}<c<1$. Then we have
\begin{align*}
0\leq U\at{x}\leq C_c \,a_c \at{x}
\end{align*}
for some constant $C_c$ depending on $c$ and all $x\in\Rset$, where $a_c\in\calC$
is defined in \eqref{lem:decay1.eqn1}.
\end{corollary}
\begin{proof}
We rewrite \eqref{eq:nlEv2} as
\begin{align}
  \notag%\label{eq:prThdecay1}
    U= (b \ast b) \ast \bat{ c\, U + \si^{-1}f(U)- c\,U}.
\end{align}
and conclude from Lemma \ref{lem:decay1} that
\begin{align*}
  U=A_c\bat{ \si^{-1}f(U)- c\,U}\leq A_c \, \ol{U}=a_c\ast \ol{U}\,,
\end{align*}
where the function
\begin{align*}
\ol{U}(x) := \max\Bigl(0,\si^{-1}f(U(x))- c\,U(x)\Bigr)
\end{align*}
has compact support due to the \BMHR superlinearity \EMHR of $f$ in \eqref{ass:f1} and the unimodality of $U$. The claim now follows from elementary properties of convolution integrals.
\end{proof}
Corollary \ref{Corr:ExpDecay} ensures that the nonlinear eigenfunction $U=b\ast V$ decays as fast as the function $a_c$ and a similar statement holds for $V = b\ast f\at{U}/\si$ due to the smoothness of $f$ at the origin. However, the details of the decay depend on the choice of $c$ and on the properties the kernel $b$, especially on the regularity of its Fourier transform:
\begin{enumerate}
\item
Algebraic decay of $a_c$ can be deduced from the real differentiability of $\widehat{b}$. For instance, we have
\begin{align*}
\sup_{x\in\Rset} \babs{x^{m} a_c\at{x}}\leq C \int\limits_\Rset
\Babs{\frac{\dint }{\dint k^{m}}\frac{\widehat{b}^2\at{k}}{1-c\,\widehat{b}^2\at{k}}}\dint{k}
\leq C \nnorm{\widehat{b}}^2_{\fspaceH^m\at\Rset}
\end{align*}
thanks to \eqref{eq:decaylemFour} and the estimate in \eqref{lem:decay1.eqn2}.
\item
Paley-Wiener theory, see for instance \cite[section VI.7]{Kat04}, relates the exponential decay of $a_c$ to the existence of holomorphic extensions.  In particular, if $\widehat{b}$ can be extended to a holomorphic function on a strip around the real axis, then $\widehat{a}_c$ has the same property but the width of the maximal strip might depend on the choice of $c$ due to complex zeros of the function $1-c\,\widehat{a}$.
\item
The qualitative decay properties in Corollary \ref{Corr:ExpDecay} improve if $c$ approaches the lower bound $\si^{-1}f^\prime\at{0}$, but the constant $C_c$ might explode in this limit.
\end{enumerate}
We finally mention that the expected decay rate of $V$ can -- at least for sufficiently nice kernels $b$ --  be characterized heuristically as follows: The exponential ansatz $V\at{x}\approx C\exp\at{-\la x}$ for $x\to\infty$ implies
\begin{align*}
\bat{b\ast V}\at{x}\approx \at{C \int\limits_\Rset b\at{x}\exp\at{\la y}\dint{y}}\exp\at{-\la\,x }
\end{align*}
and linearizing \eqref{eq:nlEv} in the tail we find
\begin{align*}
M\at{\la}=
\frac{\si}{f^\prime\at{0}}\,,
\qquad
M\at{\la}:=
 \int\limits_\Rset a\at{y}\exp\at{\la y}\dint{y}=
\at{\int\limits_\Rset b\at{y}\exp\at{\la y}\dint{y}}^2
\end{align*}
as a transcendental equation for the decay rate $\la$. Notice that the function $M$  involves exponential moments of $a$, is strictly increasing for positive arguments, and might blow up at a finite value of $\la$.

%
%
%----------------------------------------------------------------
\section{Scaling Limits}
\label{sect:scaling} %
%----------------------------------------------------------------
%
%
%----------------------------------------------------------------
\subsection{KdV limit for small eigenvalues}
\label{sect:KdV} %
%----------------------------------------------------------------
%
%
In this section we discuss a first asymptotic scaling limit $\eps\to0$ for the solutions to the nonlinear eigenvalue problem \eqref{eq:nlEv}. In this regime, the eigenvalue \BMHR $\si$ \EMHR is slightly above the critical value \BMHR $\al=f^\prime\at{0}$ \EMHR via
\begin{align*}
\si-\al\sim \eps^2
\end{align*}
and the energetic terms $\calK\at{V}$ and $\calP\at{V}$ are both proportional to $\eps^3$. For FPUT chains, this asymptotic regime is usually called the KdV limit and regards traveling waves that propagate with near sonic speed and have small amplitudes but large wave \BMHR lengths. \EMHR The key idea is that the profile functions $U$ and $V$ converge after a suitable rescaling to the solitary \BMHR wave of \EMHR a certain KdV equation. The latter is the homoclinic solution to the planar Hamiltonian ODE
\begin{align}
\label{Eqn:KdVODE}
\ol{U}^{\prime\prime} =\ka_1 \,\ol{U}-\ka_2\,\ol{U}^2
\end{align}
where the positive coefficients $\ka_1$, $\ka_2$ depend on the constants $\al$ and $\beta$ from \eqref{ass:f1}, i.e., on the first two derivatives of $f$ at the origin.
\par
The relation between traveling waves in FPUT lattices and KdV equations was first observed in \cite{ZK65} and has later been made rigorous in \cite{FP99,FML15}. Generalizations to more complex atomic system can be found in \cite{HML16, HW17, ChH18, HM19a}. We also emphasize that the KdV equation does not only govern an asymptotic regime of lattice waves but is rather a universal modulation equation for a broad range of nonlinear dispersive systems. We refer to \cite{SW00,HW09,CCPS12,GMWZ14} in the context of initial value problems in Hamiltonian lattices and to \cite{Bri13,SU17} for an overview and related results for nonlinear PDEs.
%
%
%----------------------------------------------------------------
\paragraph{Heuristics}  %
%----------------------------------------------------------------
%
To illustrate the key asymptotic ideas we start with a formal derivation of the limit equation and assume that
\begin{align*}
  \sigma=\alpha+\delta^2
\end{align*}
holds for a small parameter $\delta>0$. Rescaling $U=b\ast{V}$ by
\begin{align*}
U\at{x}=\delta^2 \tilde{U}\at{\widetilde{x}}\,,\qquad \tilde{x}=\delta x
\end{align*}
we find  the formal expansion
\begin{align*}
\bat{a\ast U}\at{\delta^{-1}\tilde{x}}=\int\limits_\Rset a\at{y} \delta^2\, \tilde{U}^2\bat{\tilde{x}+\delta{y}}\dint{y}=\delta^2
\tilde{U}\at{\tilde{x}}+m\,\delta^4\, \tilde{U}^{\prime\prime}\at{\tilde{x}}+O\bat{\delta^6}
\end{align*}
with $ m :=\tfrac12 \int_\Rset y^2 a\at{y}\dint{y} $, where we used that $U$ and $a$ are even as well as $\int_\Rset a\at{y}=1$. Moreover, a similar formula holds for $\bat{a\ast U^2}\at{\delta^{-1}\tilde{x}}$ while \BMHR the Taylor expansion of $f$ around $0$ \EMHR yields
\begin{align*}
f\bat{U\nat{\delta^{-1}\tilde{x}}} = \al\,\delta^2 \tilde{U}\bat{\tilde{x}}+\be\,\delta^4\,\tilde{U}^2\at{\tilde{x}}+O\bat{\delta^4}\,.
\end{align*}
Inserting all asymptotic formulas into \eqref{Cor:Existence} and diving by $\delta^2$ we finally get
\begin{align*}
\tilde{U}=m\,
\tilde{U}^{\prime
\prime}+\beta\, \tilde{U}^2+O\bat{\delta^2}\,.
\end{align*}
This is in fact the ODE for the KdV wave and implies $\calP\at{V}\sim\delta^3$ as well as $\calK\at{V}\sim\delta^3$.
\par
Notice, however, that the correction terms in the above asymptotic expansion of the convolution integral involve higher derivatives of $\tilde{U}$. A rigorous justification of our arguments is hence not straight forward but requires a careful analysis of nonlinear fixed point problems built of singularly perturbed pseudo-differential operators. This has been done in \cite{FP99, HML16,HM19a} in the context of Hamiltonian lattice waves and the underlying arguments can also be applied to \eqref{eq:nlEv2} for a wide class of kernel functions $a$.  In this paper we present a different approach which is based on the variational existence theory from \S\ref{sect:Var}. In particular, the small scaling parameter is no longer defined via $\si$ but in terms of the norm constraint.
%
%
%----------------------------------------------------------------
\paragraph{Variational approach} %
%----------------------------------------------------------------
%
We consider a family of solutions $\triple{V_\eps}{U_\eps}{\si_\eps}$ with
\begin{align*}
V_\eps\in\calC\,,\qquad \calK\bat{V_\eps}=\eps^3\,,\qquad \calP\bat{V_\eps}=P\bat{\eps^3}\,,\qquad U_\eps=b\ast V_\eps
\end{align*}
as provided by Corollary \ref{Cor:Existence}, where  $0<\eps<1$  is the small parameter. Motivated by the heuristic considerations we introduce the rescaled profile functions $\ol{V}_\eps$ via
\begin{align}
\label{Eqn:KdVScaling}
V_\eps\at{x}=\eps^2\ol{V}_\eps\at{\ol{x}}\,,\qquad \ol{x}=\eps x
\end{align}
and we aim to show that $\ol{V}_\eps$ converge as $\eps\to0$ to a unique limit profile \BMHR $\ol{V}_0=\ol{U}_0$. \EMHR To this end we restate the eigenvalue problem \eqref{eq:nlEv} as
\begin{align}
\label{eq:nleveps}
\si_\eps \ol{V}_\eps=\ol{b}\ast\bat{\alpha\,\ol{b}\ast \ol{V}_\eps+\eps^2 \beta \,\ol{b}\ast \at{\ol{V}_\eps}^2+\eps^{-2}R\bat{\eps^2\, \ol{b}\ast \ol{V}_\eps}}
\end{align}
\BMHR with \EMHR %
\begin{align*}
\ol{b}\at{\ol{x}}:=\eps^{-1} b\at{\eps^{-1}\ol{x}}\,,\qquad
\widehat{\ol{b}}\at{\ol{k}}=\widehat{b}\at{\eps \ol k }\,,
\end{align*}
\BMHR where \EMHR $R$ defined by
\begin{align*}
R\at{r}=f\at{r}-\al\,r-\tfrac12\,\beta\,r^2
\end{align*}
represents the cubic and higher order nonlinearities. In consistency with \eqref{Eqn:KdVScaling} we also set
\begin{align*}
U_\eps\at{x}=\eps^2\ol{U}_\eps\at{\ol{x}}
\end{align*}
and find $\ol{U}_\eps=\ol{b}\ast\ol{V}_\eps$.
\bigpar
We next specify our \BMHR refined \EMHR assumptions on the kernel function \BMHR $b$ \EMHR and establish some elementary but useful a priori estimates.
\begin{assumption}[additional assumptions for KdV limit]\label{Ass:KdV}
Besides Assumption \ref{Ass:MainKernel} we suppose that $\widehat{b}$ is of class $\fspaceC^2$ in some open neighborhood of $k=0$ and that there exists a constant $ C>0$ such that
\begin{align}
\label{Ass:KdVKernel}
\babs{\hat{b}\at{k}}&\leq \sqrt{\frac{1}{1+C k^2}}\,,\qquad\qquad
|\hat{b}\at{k}|^2\geq 1-C k^2
\end{align}
holds for all  $\BMHR k \in \EMHR \Rset$.
\end{assumption}
\begin{lemma}[simple a priori estimates]\label{Lem:KdvEst}
We have
\begin{align}
\notag%\label{eq:lem:ap}
P(\eps^3) &\geq  \alpha\eps^3+c\,\eps^5\,,\qquad
\si_\eps \geq \alpha+c\,\eps^2\,.
\end{align}
for some constant $c>0$ independent of $\eps$.
\end{lemma}
\begin{proof}
We fix a smooth and rapidly decaying function $\ol{V}\in\calC$ with $\tfrac12\bnorm{\ol{V}}_2^2=1$ and estimate the contribution to the potential energy for the family
\begin{align*}
\widetilde{V}_\eps\at{x} : = \eps^2 \ol{V}\at{\eps x}\,.
\end{align*}
Using the Plancherel identity \eqref{Eqn:Plancherel} and setting $
\widetilde{U}_\eps:=b\ast \widetilde{V}_\eps$ we find
\begin{align*}
\calQ\bat{\widetilde{V}_\eps}&=\frac{\al}{2} \int\limits_\Rset \widetilde{U}_\eps\at{x}^2\dint x
  =%
\frac{\al}{2} \,\eps^3 \int\limits_\Rset (\ol{b} \ast \ol{V})^2\at{\ol{x}}\dint \ol{x}=
\frac{\al}{4\pi} \,\eps^3 \int\limits_\Rset \widehat{b}\at{\eps\ol{k}}^2\widehat{\ol{V}}\at{\ol{k}}^2\dint \ol{k}
\\&\geq
\frac{\al}{4\pi} \,\eps^3\int\limits_\Rset\at{1-C_1\,\eps^2\,\ol{k}^2}\widehat{\ol{V}}\at{\ol{k}}^2\dint \ol{k}\geq \frac{\al}{2}\,\eps^3\Bat{\bnorm{\ol{V}}_2^2-C_1\,\eps^2\bnorm{\ol{V}^\prime}_2^2}
\\&\geq \al\,\eps^3 -\frac{\al}{2} \,C_1\,\eps^5\,\bnorm{\ol{V}^\prime}_2^2\,,
\end{align*}
where the constant $C_1$ is provided by \eqref{Ass:KdVKernel}. On the other hand, the \BMHR properties $f$ imply \EMHR
\begin{align*}
\calP\bat{\widetilde{V}_\eps}-\calQ\bat{\widetilde{V}_\eps}\geq c_2
\int\limits_\Rset\bat{\widetilde{U}_\eps\at{\tilde{x}}}^3\dint{x}
=c_2\,\eps^5\, \int\limits_\Rset\bat{ \bat{\bar{b}\ast\bar{V}}\at{\bar{x}}}^3\dint{\bar{x}}
\end{align*}
for some (small but positive) constant $c_2>0$, and since $\ol{b}\ast \ol{V}$ converges to $\ol{V}$ strongly in $\fspaceL^3\at\Rset$ as $\eps\to0$ we find
\begin{align*}
\calP\bat{\widetilde{V}_\eps}-\calQ\bat{\widetilde{V}_\eps}\geq \frac {c_2}{2}\,\eps^5\, \int\limits_\Rset \bat{\bar{V}\at{\bar{x}}}^3\dint{\bar{x}}
\end{align*}
for all sufficiently small $\eps$. Since $C_1$ and $c_2$ do not depend on the choice of $\ol{V}$, we can replace $\ol{V}$ via
\begin{align*}
\ol{V}\at{\ol{x}}\qquad\rightsquigarrow\qquad \la^{1/2}\ol{V}\at{\la \ol{x}}
\end{align*}
by a suitable dilation of itself to guarantee the estimate
\begin{align*}
C_1
\int\limits_\Rset \bat{\bar{V}^\prime\at{\bar{x}}}^2\dint{\bar{x}}\leq
\frac{c_2}{2\,\alpha}\int\limits_\Rset \bat{\bar{V}\at{\bar{x}}}^3\dint{\bar{x}}\,,
\end{align*}
and this implies the first claim via
\begin{align*}
P\bat{\eps^3}\geq \calP\bat{\widetilde{V}_\eps}\geq \al\,\eps^3+\frac{c_2}{4}\,\eps^5\,.
\end{align*}
Finally, the second \BMHR assertion \EMHR follows via $\si_\eps \geq \eps^{-3}P\bat{\eps^3}\geq \eps^{-3}\calP\bat{\widetilde{V}_\eps}$ from \eqref{Cor:Existence.Eqn1}.
\end{proof}
Our main result in this section establishes the convergence of the rescaled profile functions using nonlinear compactness arguments as well as the uniqueness of accumulation points.
\begin{theorem}[convergence of maximizers]\label{prop:kdv}
We have
\begin{align}
\label{prop:kdv.Eqn0}
\frac{\si_\eps-\al}{\eps^2}\quad\xrightarrow{\;\;\eps\to0\;\;}\quad
\frac{\beta\,\ka_1}{\ka_2}
\end{align}
as well as
\begin{align}
\label{prop:kdv.Eqn1}
\ol{V}_\eps\quad\xrightarrow{\;\;\eps\to0\;\;}\quad \ol{U}_0\,,\qquad
\ol{U}_\eps\quad\xrightarrow{\;\;\eps\to0\;\;}\quad \ol{U}_0
\end{align}
strongly in $\fspaceL^2\at\Rset$, where the limit is given by
\begin{align*}
\ol{U}_{0} (\ol{x})= \frac{3\,\ka_1}{2\,\ka_2}\, \mathrm{sech}^2\at{\frac{\sqrt{\ka_1}}{2}\, \ol{x}}
\end{align*}
and equals the homoclinic solution of \eqref{Eqn:KdVODE}. The values of the constants $\ka_1$, $\ka_2$ are given in  the proof, see \eqref{prop:kdv.peqn0} and \eqref{prop:kdv.peqn2}.
\end{theorem}
\begin{proof}
Fourier transforming the scaled Euler-Lagrange equation \eqref{eq:nleveps} we find
\begin{align}
\label{eq:EULAtra}
\widehat{\ol{U}}_\eps\at{\ol{k}}=\widehat{b}\at{\eps \ol{k}}
\widehat{\ol{V}}_\eps\at{\ol{k}}\,,\qquad
\widehat{\ol{V}}_\eps\at{\ol{k}}=\frac{\eps^2\widehat{b}\bat{\eps \ol{k}}}{\si_\eps-\alpha\,\widehat{b}^2\bat{\eps\ol{k}}} \left(\beta\,\widehat{\;\ol{U} _\eps^2\;}\at{\ol{k}}+ \eps^{-4}\widehat{R\bat{\eps^2 \ol{U}_\eps}\at{\ol{k}}}\right)\,.
\end{align}
\emph{\ul{A priori estimates}}\,: %
We have
\begin{align*}
\bnorm{ \widehat{\ol{U}_\eps^2\;}}_\infty \leq
\bnorm{ \ol{U}_\eps^2\;}_1= \bnorm{ \ol{b}\ast\ol{V}_\eps\;}_2^2\leq \bnorm{\ol{b}}_1^2\,\norm{\ol{V}_\eps}_2^2\leq \bnorm{\ol{V}_\eps}_2^2\leq 2
\end{align*}
and using
\begin{align*}
0\leq {R}\at{s}\leq C s^3\quad\text{for}\quad 0\leq s\leq 1\,,\qquad \bnorm{\eps^2\bar{U}_\eps}_\infty\leq C\eps^2\bnorm{\bar{b}}_2\,\bnorm{\ol{V}_\eps}_2\leq C \eps^{3/2}\leq 1
\end{align*}
we can estimate the remainder terms by
\begin{align}
\label{Eqn:KdVRemainderEstimates}
\bnorm{\eps^{-4}\widehat{R\at{\eps^2 \bar{U}_\eps}}}_\infty
\leq C \eps^{-4} \bnorm{R\at{\eps^2 \bar{U}_\eps}}_1\leq C\eps^{-4}\bnorm{\eps^2\bar{U}_\eps}_\infty \bnorm{\eps^2\bar{U}_\eps}_2^2\leq C\eps^{3/2}\norm{\ol{b}}_1^2\bnorm{\bar{V}_\eps}_2^2\leq C\eps^{3/2}\,.
\end{align}
Our assumptions on the kernel function $b$ als well as the lower bound for $\si$ -- see Assumption \ref{Ass:KdV} and Lemma \ref{Lem:KdvEst} --  imply
\begin{align*}
\si_\eps - \al\,\widehat{b}^2\bat{\eps\bar{k}}>0\qquad\text{for all}\quad \bar{k}\in\Rset
\end{align*}
as well as
\begin{align}
\label{Eqn:KdVTightness1}
\abs{\frac{\eps^2\,\widehat{b}\bat{\eps \ol{k}}}{\si_\eps-\alpha\widehat{b}^2\bat{\eps\ol{k}}}} \leq C\frac{\eps^2\sqrt{\frac{\D 1}{\D 1+c\eps^2\ol{k}^2}}}{\alpha+c\eps^2-\frac{\D \alpha}{\D 1+c\eps^2\ol{k}^2}}\leq C \frac{\eps^2\sqrt{1+c\eps^2\ol{k}^2}}{(\alpha+ c\eps^2) \at{1+c\eps^2\ol{k}^2}-\alpha}
\leq C \frac{\sqrt{1+c\eps^2\ol{k}^2}}{1+\ol{k}^2}\,.
\end{align}
In summary, we arrive at
\begin{align}
\notag% \label{eq:vhatpw}
\babs{\widehat{\ol{V}}_\eps\at{\ol{k}}}\leq  \frac{C \sqrt{1+c\eps^2\ol{k}^2}}{c+\ol{k}^2}\,,\qquad
\babs{\widehat{\ol{U}}_\eps\at{\ol{k}}}\leq  \frac{C}{c+\ol{k}^2}\,,
\end{align}
which provides the validity of
\begin{align}
\notag%\label{eq:barUH1}
\bnorm{\ol{U}_\eps}_{1,2}^2&= 2\pi\int\limits_{-\infty}^{+\infty} \bat{1+\ol{k}^2}\babs{ \widehat{\ol{U}}_\eps(\ol{k})}^2 \dint{\ol{k}} \leq C
\end{align}
and
\begin{align}
\label{eq:barULinf}
\bnorm{\ol{U}_\eps}_\infty & \leq 2\pi \bnorm{\widehat{\ol{U}}_\eps}_1\leq C
\end{align}
for all sufficiently small $\eps$.
\par
\emph{\ul{Convergent subsequences}}\,:
By standard compactness results on bounded intervals we choose a subsequence such that $\ol{U}_\eps$ converges as $\eps\to0$ strongly in $\fspaceL^2_\loc\at\Rset$ as well as pointwise to some limit $\ol{U}_0$. Since the unimodality of $\ol{U}_\eps$ combined with \eqref{eq:barULinf} implies
\begin{align*}
0\leq \ol{U}_\eps\at{\ol{x}}\leq \min\left\{\frac{\bnorm{\ol{U}_
\eps}_2}{2\sqrt{x}}, \bnorm{\ol{U}_\eps}_\infty\right\} \leq \frac{C}{\sqrt{1+\ol{x}}}
\end{align*}
we find
\begin{align*}
0\leq \ol{U}^2_\eps\at{\bar{x}}\leq \frac{C}{1+\bar{x}}
\end{align*}
for all $\bar{x}\in\Rset$. By the Dominated Convergence Theorem and the Plancherel identity we then conclude that
\begin{align*}
\ol{U}_\eps^2\quad\xrightarrow{\;\eps\to0\;}\quad\ol{U}_0^2
\qquad\text{and}\qquad\widehat{\ol{U}_\eps^2} \quad\xrightarrow{\eps\to0\;}\quad
\widehat{\ol{U}_0^2}\qquad \text{both strongly in $\fspaceL^2\at\Rset$}\,.
\end{align*}
Passing to a further subsequence we can assume \BMHR -- see for instance \cite[Theorem 4.9]{Bre11} -- \EMHR that the latter convergence holds also pointwise almost everywhere with respect to $\ol{k}\in\Rset$ and that
\begin{align*}
\widehat{\ol{U}^2_\eps}\at{\bar{k}}\leq \ol{H}\at{\bar{k}}
\end{align*}
is satisfied for almost all $\bar{k}$ and a dominating function $\ol{H}\in\fspaceL^2\at{\Rset}$ independent of $\eps$. Finally, extracting a further subsequence we can assume that
\begin{align*}
d_\eps : = \frac{\si_\eps-\al}{\eps^2}\quad \xrightarrow{\;\eps\to0\;}\quad d_0>0\,,
\end{align*}
where $d_0$ might here be finite or infinite.
\par%
\emph{\ul{Passage to the limit and uniqueness of accumulation points}}\,: %
Our assumptions concerning the regularity of $\ol{b}$ at the origin imply
\begin{align}
\notag%\label{eq:EULAsymbol}
\widehat{b}\bat{\eps \ol{k}}\quad\xrightarrow{\eps\to0\;}\quad1\qquad \text{and}\qquad
\frac{\eps^2\widehat{b}\bat{\eps \ol{k}}}{\si_\eps-\alpha\,\widehat{b}^2\bat{\eps\ol{k}}}  \quad\xrightarrow{\eps\to0\;}\quad \frac{1}{d_0+\alpha\, \widehat{b}''(0)\, \ol{k}^2}
 \end{align}
for almost all $\ol{k}\in\Rset$. From the Dominated Convergence theorem, \BMHR equation \eqref{eq:EULAtra}, \EMHR and the existence of the $\fspaceL^2$-majorant $\ol{H}$ we then infer that
\begin{align*}
\widehat{\ol{U}}_\eps, \widehat{\ol{V}}_\eps\quad\xrightarrow{\eps\to0\;}\quad\frac{\widehat{\beta\,\bar{U}_0}^2}{d_0+ \alpha\, \widehat{b}''(0)\, \ol{k}^2} \qquad \text{strongly in $\fspaceL^2\at\Rset$}
\end{align*}
along the chosen subsequence, where the contributions from the remainder terms vanish due to the $\fspaceL^\infty$-bounds from \eqref{Eqn:KdVRemainderEstimates} and the $\fspaceL^2$-majorants in \eqref{Eqn:KdVTightness1}. The strong $\fspaceL^2$-convergences $\ol{U}_\eps\to\ol{U}_0$  and $\ol{V}_\eps\to\ol{U}_0$, which follow by the Plancherel Theorem, guarantee the validity of
\begin{align*}
d_0\,\ol{U}_0+\alpha\, \widehat{b}''(0) \,\ol{U}_0^{\prime\prime}=\beta\,\bar{U}_0^2\,.
\end{align*}
Thanks to $\beta>0$ and $ \widehat{b}''(0)<0$, this planar
ODE equals \eqref{Eqn:KdVODE}  with
\begin{align}
\label{prop:kdv.peqn0}
\ka_1=\frac{d_0}{\alpha\,\babs{\widehat{b}^{\prime\prime}\at{0}}}\,,\qquad
\ka_2=\frac{\beta}{\alpha\,\babs{\widehat{b}^{\prime\prime}\at{0}}}
\end{align}
and admits the homoclinic solution
\begin{align}
\label{prop:kdv.peqn1}
 \ol{U}_{0} (\ol{x})= \frac{3 \, d_0}{2 \,\beta} \mathrm{sech}^2\at{
\sqrt{\frac{d_0}{4\,\alpha\,\babs{\widehat{b}''(0)}} \,}
\bar{x}}\,,
\end{align}
which is unique within the cone $\calC$. By construction, and due to the strong convergence $\ol{V}_\eps\to \ol{U}_0$, we further have $\nnorm{\ol{U}_0}_2=2$, so a direct computations yields $\ka_1^{3/2}\ka_2^{-2}=1/3$ and hence
\begin{align}
\label{prop:kdv.peqn2}
d_0 = \frac{\beta^{4/3}}{3^{2/3}\al^{1/3}\babs{\widehat{b}^{\prime\prime}\at{0}}^{1/3}}\,.
\end{align}
In summary, we have shown that there exists precisely one accumulation point for $\ol{U}_\eps$ and this finally implies the desired convergence result $\ol{U}_\eps\to\ol{U}_0$. The remaining \BMHR convergence statements in \eqref{prop:kdv.Eqn0} and \eqref{prop:kdv.Eqn1} \EMHR follow immediately.
\end{proof}
The arguments in the proof of Theorem \ref{prop:kdv} can also be used to show that  $\ol{U}_\eps^{\prime}$ converges to $\ol{U}_0^\prime$ strongly in $\fspaceL^2\at\Rset$. More generally, by bootstrapping arguments we can establish \eqref{prop:kdv.Eqn1} with respect to higher Sobolev norms provided that the kernel $b$ is sufficiently smooth.
%
%
%----------------------------------------------------------------
\subsection{The limit of large eigenvalues}
\label{sect:HE} %
%----------------------------------------------------------------
%
%
If the nonlinearity $f$ exhibits an algebraic singularity, there exists another asymptotic regime  related to large values of $\si$ and $\calP\at{V}$. The heuristic idea is that
\begin{align}
\label{Eqn:DefN}
N:=\si^{-1} f\at{U}
\end{align}
becomes asymptotically a Dirac distributions with finite mass centered around $0$ so that $V$ and $U$ can be approximated by certain multiples of $b$ and $a=b\ast b$, respectively. For FPUT chains, this limit is usually called the high-energy limit and was first studied in  \cite{FM02,Tre04}. The underlying asymptotic analysis has later been refined by the authors in \cite{HM15,Her17,HM19b} to prove the orbital stability of high-energies waves  in chains with Lennard-Jones-type potentials.
\par
In what follows we study the prototypical nonlinearity \eqref{eq:fsing} and present a rather simple proof for the convergence of the profile functions in the variational setting. This first result, however, neither provides convergence rates nor explicit scaling relations for the eigenvalue $\si$. We therefore continue with a refined asymptotic analysis for smooth convolution kernels and characterize the fine structure of the aforementioned approximate Dirac distribution in greater detail. In this way we derive a scaling law for the eigenvalue $\si$, which differs significantly from the corresponding law for FPUT chains
as these come with a much less regular kernel function $a=b\ast b$.
\begin{assumption}[assumptions for limit]
\label{Ass:HE}
We suppose that $f$ is given by \eqref{eq:fsing} and that $a=b\ast b$ is sufficiently smooth so that both $a$ and $a^{\prime\prime}$ belongs to $\fspace{BC}\at\Rset\cap\fspaceL^1\at\Rset$.
\end{assumption}
We emphasize that the variational existence result from \S\ref{sect:Var} can also be applied to the nonlinearity \eqref{eq:fsing} although it is not defined on $\cointerval{0}{\infty}$ as required by Assumption \ref{Ass:MainNonl} but only on $\cointerval{0}{1}$. In fact, Young's inequality
\begin{align*}
\nnorm{b\ast V}_\infty \leq \nnorm{b}_2\nnorm{V}_2=\nnorm{b}_2\sqrt{2\,\calK\at{V}}
\end{align*}
reveals that the assertions of Propositions \ref{prop:PT} and \ref{lem:ex1} as well as Corollary \ref{Cor:Existence} remain valid as long as the norm parameter $K$ is confined by
\begin{align}
\label{Eqn:KMax}
0<K<K_{\max}:=
\frac{1}{2\,\norm{b}_2^2}=\frac{1}{2 \, a\at{0}}\,.
\end{align}
%
%
%
%
%----------------------------------------------------------------
\paragraph{Convergence of profiles} %
%----------------------------------------------------------------
%
Heuristic arguments as well as numerical simulation as in the third panel of Figure \ref{fig:waves} indicate the following asymptotic result: If the maximum of $U$ is close the the singular value $1$ of $f$, the profile $N$ from \eqref{Eqn:DefN} concentrates near zero and the profiles $U$ and $V$ approach limit functions $V_0$ and $U_0=b\ast V_0$ with
\begin{align*}
  V_0 \at{x} := \frac{b\at{x}}{a\at{0}}\,,\qquad U_0\at{x}=\frac{a\at{x}}{a\at{0}}\,.
\end{align*}
To prove this in the variational framework from \S\ref{sect:Var}, we introduce a small parameter $0<\delta<1$ and consider a family $\triple{V_\delta}{U_\delta}{\si_\delta}$ of solution to the nonlinear eigenvalue problem \eqref{eq:nlEv} with
\begin{align*}
  V_\delta\in\calC\,,\qquad \calK\at{V_\delta}=\at{1-\delta}K_{\max}\,,\qquad \calP\at{V_\delta}=P\bat{\at{1-\delta}K_{\max}}\,,\qquad U_\delta=b\ast V_\delta
\end{align*}
as provided by Corollary \ref{Cor:Existence}, see also the comment to \eqref{Eqn:KMax}. We finally introduce
\begin{align*}
 \eps_\delta :=1- U_{\delta}\at{0}
\end{align*}
and notice that the unimodality of $U_\delta$ ensures $\nnorm{U_\delta}_\infty=1-\eps_\delta$.
\begin{theorem}[convergence result]\label{thm:coarse}
For $\delta\to0$ we have
\begin{align*}
V_\delta \quad\xrightarrow{\;\delta\to0\;}\quad V_0 \,,\qquad \qquad
U_{\delta} \quad\xrightarrow{\;\delta\to0\;}\quad {U}_0
\end{align*}
both  strongly in $\fspaceL^2\at\Rset$  as well as $\eps_\delta\to0$ and $\si_\delta\to\infty$.
\end{theorem}
\begin{proof}
\emph{\ul{Convergence of $\eps_\delta$}}\,:
Notice that
\begin{align}
  \calP\at{V_\delta}\geq \calP\bat{\sqrt{1-\delta}V_0
}&=  \int\limits_\Rset \frac{\dint{x}}{m\left(1- \sqrt{1-\delta}\D \frac{a\at{x}}{a\at{0}} \right)^m} \quad\xrightarrow{\;\delta\to0\;}\quad \infty  \label{eq:PVde}
\end{align}
holds by construction and \BMHR Assumption \ref{Ass:HE}. Now \EMHR assume for contradiction that $\eps_\delta$ does not converge to zero. Then
there exists a constant $\eps_0>$ such that $\eps_\delta\geq\eps_0$ holds along a fixed subsequence with $\delta\to0$ and $K_\delta\to K_0=K_{\max}$. \BMHR In particular, we have $\norm{U_\delta}_\infty\leq 1-\eps_0$ \EMHR and this uniform distance to the singularity of $F$  guarantees
\begin{align*}
F\bat{U_\delta\at{x}}  \leq C U_\delta^2\at{x}
\end{align*}
for some constant $C$ and all $x \in \Rset$.  \BMHR We thus get \EMHR
\begin{align*}
\calP\at{V_\delta}\leq C\norm{U_\delta}_2^2\leq C\bnorm{b}_1^2\bnorm{V_\delta}_2^2\leq C
\end{align*}
\BMHR and hence a \EMHR contradiction to \eqref{eq:PVde}.
\par
\emph{\ul{Convergence of profiles and $\si_\delta$}}\,:
Observing $U_\delta\at{0}=\skp{V_\delta}{b}$ we calculate
\begin{align*}
\bnorm{ V_{\delta} -{V}_0}_2^2& =
\bnorm{V_{\delta}}_2^2+ \bnorm{{V}_0}_2^2 - 2 \bskp{V_{\delta}}{{V}_0} = 2K_\delta+2K_0-2\norm{b}_2^{-2} U_\delta\at{0}
\\&=
\norm{b
}_2^{-2}\at{1 - \delta +1 - 2+2 \eps_\delta}
\end{align*}
and in combination with Young's  inequality
\begin{align*}
\bnorm{U_{\delta} -{U}_0}_\infty\leq \bnorm{b}_2
\bnorm{V_{\delta} -{V}_0}_2
\end{align*}
we obtain the convergence results for both $V_\delta$ and $U_\delta$. Finally, \eqref{Cor:Existence.Eqn1} gives
\begin{align*}
 \si_\delta \geq \frac{\calP(V_\delta)}{\calK\at{V_\delta}}
\end{align*}
so the claim on $\si$ follows from \eqref{eq:PVde}.
\end{proof}
\BMHR%
%\begin{remark}
The proof of Theorem \ref{thm:coarse}
resembles key arguments from \cite{FM02} and uses only
Assumption \ref{Ass:MainKernel}.  The next step, however, requires the higher regularity properties formulated in Assumption \ref{Ass:HE} and does not apply to piecewise linear kernel functions.
\EMHR
%
%----------------------------------------------------------------
\paragraph{Asymptotic analysis of wave speed} %
%----------------------------------------------------------------
%
We finally derive the scaling relation between $\si_\delta$ and $\eps_\delta$.
\begin{proposition}[asymptotics of eigenvalue]
\label{Prop:ScalingLaws}
We have
\begin{align}
\label{Prop:ScalingLaws.Eqn1}
\eta_\delta:= \si_\delta \eps_\delta^{m+1/2}
\quad
\xrightarrow{\;\delta\to0\;}\quad
\eta_0:=
\frac{\sqrt{2\pi}\sqrt{a\at{0}}^3}{\sqrt{\abs{a^{\prime\prime}\at{0}}}}\frac{\Ga\at{m+\tfrac12}}{\Ga\at{m+1}}\,,
\end{align}
where $\Ga$ denotes the Gamma function.
\end{proposition}
\begin{proof}
\emph{\ul{Preliminaries}}\,:
Denoting the indicator function of the interval $\ccinterval{-1}{+1}$ by $\chi$, we define
\begin{align*}
N_{1,\delta}:=\si_\delta^{-1} f\bat{U_\delta} \chi\,,\qquad
N_{2,\delta}:= \si_\delta^{-1} f\bat{U_\delta} \at{1-\chi}
\end{align*}
and split  $U_\delta=U_{1,\delta}+U_{2,\delta}$ via
\begin{align*}
U_{1,\delta} : = a\ast N_{1,\delta}\,,\qquad
U_{2,\delta} : =  a\ast N_{2,\delta}\,.
\end{align*}
This implies
\begin{align}
\notag%\label{eq:bV}
  U_{i,\delta}^{\prime  \prime} =a^{\prime\prime}\ast N_{i,\delta}
\end{align}
for $i=1,2$ and our first goal is to establish the convergence of these second derivatives.
\par
\emph{\ul{Improved convergence result}}\,: %
Theorem \ref{thm:coarse} \BMHR and the unimodality of $U_\delta$ imply \EMHR that $\sup_{\abs{x}\geq 1} U_\delta\at{x}<1$ holds uniformly with respect to sufficiently small $\delta$. We thus find
\begin{align*}
\bnorm{N_{2,\delta}}_\infty\leq C \si_\delta^{-1}\bnorm{U_\delta\at{1-\chi}}_\infty\quad
\xrightarrow{\;\delta\to0\;}\quad0
\end{align*}
thanks to \BMHR $\si_\delta\to\infty$ \EMHR and obtain
\begin{align}
\label{Prop:ScalingLaws.PEqn1}
\bnorm{U_{2,\delta}}_\infty \leq  C \bnorm{a}_1\bnorm{N_{2,\delta}}_\infty \quad
\xrightarrow{\;\delta\to0\;}\quad0\,,
\qquad
\bnorm{U_{2,\delta}^{\prime\prime}}_\infty \leq  C \bnorm{a^{\prime\prime}}_1\bnorm{N_{2,\delta}}_\infty  \quad
\xrightarrow{\;\delta\to0\;}\quad0\,.
\end{align}
We also have
\begin{align*}
\int\limits_{-1}^{+1} U_\delta\at{x}\dint{x}=\bskp{U_{1,\delta}+U_{2,\delta}}{ \chi}=
\bskp{N_{1,\delta}\chi}{a\ast\chi}+\bskp{N_{2,\delta}\chi}{a\ast\chi}
\end{align*}
and since $a\ast\chi$ is uniformly positive in the interval $\ccinterval{-1}{+1}$ we deduce from the convergence results for $U_\delta$ and $N_{2,\delta}$ that $\bnorm{N_{1,\delta}}_1$ is uniformly bounded. In particular, there exist subsequences for $\delta\to0$ such that $N_{1,\delta}$ converges weakly$^\ast$ to a
limit measure, and Theorem \ref{thm:coarse} combined with $0\leq U_0\at{x}<1$ for $x\neq0$ implies that this limit is a Dirac measure with finite mass $\mu$ concentrated at $x=0$. Moreover, the mass $\mu$ is uniquely determined by the asymptotic identity
\begin{align*}
\int\limits_{-1}^{+1} {U}_0\at{x}\dint{x}=\mu\, \nat{a\ast\chi}\at{0}=\mu\,\int\limits_{-1}^{+1}a\at{x}\dint{x}\,,
\end{align*}
and this uniqueness of accumulation points ensures the weak$^\star$ convergence of the entire family $N_{1,\delta}$. The smoothness of the kernel $a$ combined with Young's convolution inequality thus provides
\begin{align*}
\bnorm{U_{1,\delta}-{U}_0}_\infty+
\bnorm{U_{1,\delta}^{\prime\prime}-{U}_0^{\prime\prime}}_\infty
 \quad
\xrightarrow{\;\delta\to0\;}\quad0\,.
\end{align*}
\BMHR Moreover, in view of \EMHR  \eqref{Prop:ScalingLaws.PEqn1} we conclude that the convergence
$U_\delta\to {U}_0$ holds even with respect to the $\fspaceC^2$-topology. \BMHR Consequently, \EMHR there exists a constant $d>0$ such that
\begin{align}
\label{Prop:ScalingLaws.PEqn2}
0\leq U_\delta\at{x}\leq 1-\eps_\delta - d\,x^2
\end{align}
holds for all $\abs{x}\leq1$ and all sufficiently small $\delta$.
\par
\emph{\ul{Scaling law for $\si_\delta$}}\,:
\BMHR To describe the fine structure of $N_{1,\delta}$, we define a rescaled function $\ol{W}_\delta$ by \EMHR
\begin{align}
\label{eq:WHE}
  \ol{W}_\delta(\ol{x}):= \eps_\delta^{m+1} f\at{U_\delta\bat{\eps_\delta^{1/2} \ol{x}}} \chi\bat{\eps_\delta^{1/2} \ol{x}} =  \eps_\delta^{m+1}\si_\delta N_{1,\delta}\bat{\eps_\delta^{1/2} \ol{x}}
\end{align}
and exploit both \eqref{Prop:ScalingLaws.PEqn2} as well as \eqref{eq:fsing} to derive the uniform tightness estimate
\begin{align}\label{eq:West}
\ol{W}_\delta\at{\ol{x}} \leq\frac{1}{\at{ 1+d\, \ol{x}^2}^{m+1}}
\end{align}
for all $\bar{x}\in\Rset$.  Moreover, since $\ol{W}_\delta$ is nonnegative and unimodal its derivative is a regular measure with uniformly bounded variation (and small Dirac parts due to the jump discontinuities at $\ol{x}=\pm\eps_\delta^{-1/2}$), so \BMHR Helly's Selection Theorem -- see for instance \cite[Chapter VIII]{Nat16} -- \EMHR guarantees the existence of subsequences that converge pointwise as $\delta\to0$ to a function $\ol{W}_0$. This convergence holds even in $\fspaceL^1\at\Rset$ due to the majorant in \eqref{eq:West}, so using the smoothness of $U_\delta$ as well as the formula for $f$ in \eqref{eq:fsing} we readily verify
\BMHR the limit formula \EMHR
\begin{align*}
\ol{W}_0\at{\ol{x}}=\at{\frac{1}{1+\tfrac12\babs{ {U}_0^{\prime\prime}\at{0}}\ol{x}^2}}^{m+1}\,.
\end{align*}
In particular, the accumulation point $\ol{W}_0$ is independent of the chosen subsequence, so the entire family \BMHR $\at{\ol{W}_\delta}_\delta$ converges as $\delta\to0$ to $\ol{W}_0$. \EMHR Moreover, by construction we have
\begin{align}
\notag%\label{eq:notail2}
{U}_{1,\delta}\at{x}= \int\limits_{-\infty}^{+\infty} a\at{x-y} N_{1,\delta}\at{y}\dint{y}= \frac{\eps_\delta^{1/2} }{\sigma_\delta\, \eps_\delta^{m+1}} \int\limits_{-\eps_\delta^{-1/2}}^{+\eps_\delta^{-1/2}} a\bat{x-\eps_\delta^{1/2} \ol{y}} \ol{W}_\delta(\ol{y})
 \dint \ol{x}
\end{align}
for all $x$, and evaluating this identity for $x=0$ and in the limit $\delta\to0$ yields
\begin{align*}
\eta_\delta
\quad
\xrightarrow{\;\delta\to0\;}\quad
\eta_0:=a\at{0}\int\limits_\Rset \ol{W}_0\at{\ol{y}}\dint{\ol{y}}
\end{align*}
thanks to our convergence results for $U_{1,\delta}$ and $\ol{W}_\delta$. The desired formula for $\eta_0$ follows by computing the integral.
\end{proof}
Notice that the convergence result from Proposition \ref{Prop:ScalingLaws} might be improved as follows. Using the decay results from \S\ref{sect:decay} one can derive uniform tightness estimates for the functions $N_\delta:=\si_\delta^{-1} f\at{U_\delta}$ and this implies the approximation
\begin{align*}
N_\delta\at{x} \approx \frac{1}{\eta_0}\frac{1}{\eps_\delta^{1/2}}\ol{W}_0\at{\frac{x}{\eps_\delta^{1/2}}}\,,
\end{align*}
where the mass of the smooth Dirac on the right hand side is just $1/a\at{0}$. Such a refined analysis also reveals that $\eps_\delta$ and $\delta$ are asymptotically proportional, where the limit $\lim_{\delta\to0} \delta^{-1}\eps_\delta$ can be computed in terms of $\ol{W}_0$. 
\par
\BMHR We finally emphasize that the scaling relation between $\si_\delta$ and $\eps_\delta$ differs from the corresponding FPUT result in \cite{HM15,HM17} because here the kernel function $a$ is supposed to be smooth at the origin $x=0$, see Assumption \ref{Ass:HE}. 
The corresponding key argument in the proof of Theorem \ref{Prop:ScalingLaws}  is the uniform Taylor expansion up to quadratic order in \eqref{Prop:ScalingLaws.PEqn2}, where the curvature constant $d$ turns out to be basically independent of the small parameter $\delta$. 
 For kernels $a$ that are not twice differentiable at $x=0$
 we can no longer expect $d$ to be uniformly bounded. Instead we have to identify a spatial scaling that differs from \eqref{eq:WHE} and guarantees that the rescaled profiles $\ol{W}_\delta$ converge as $\delta\to0$ to a well-defined limit object $\ol{W}_0$. For FPUT chains, this alternative scaling has been identified in \cite{HM15}, provides an \emph{asymptotic shape ODE} that governs the behavior of $U_\delta$ near $x=0$, and implies $\sigma_\delta \sim \eps_\delta^m$ as the analogue to \eqref{Prop:ScalingLaws.Eqn1}.
\EMHR
%
%
%
%
% -----------------------------------------------------------------------------
% - appendix
% -----------------------------------------------------------------------------
%
%
\section*{Acknowledgement}
The authors gratefully acknowledge fruitful discussions with Martin Burger, Oleh Omel'chenko, Arnd Scheel, and Eric Siero.
%
%
% -----------------------------------------------------------------------------
% - Bibliography
% -----------------------------------------------------------------------------
%
%\bibliographystyle{alpha}
%\bibliography{paper}
%

\end{document}